\documentclass[a4paper,UKenglish,cleveref,autoref]{lipics-v2021}

\pdfoutput=1 

\title{A faster algorithm for finding Tarski fixed points} 

\author{John Fearnley}{Department of Computer Science, University of Liverpool}{john.fearnley@liverpool.ac.uk}{}{}

\author{D\"om\"ot\"or P\'alv\"olgyi}{MTA-ELTE Lend\"ulet Combinatorial Geometry
Research Group, Institute of Mathematics, E\"otv\"os Lor\'and University (ELTE)}{dom@cs.elte.hu}{}{}

\author{Rahul Savani}{Department of Computer Science, University of Liverpool}{rahul.savani@liverpool.ac.uk}{{https://orcid.org/0000-0003-1262-7831}}{}

\authorrunning{J. Fearnley, D. P\'alv\"olgyi, and R. Savani} 

\Copyright{John Fearnley, D\"om\"ot\"or P\'alv\"olgyi, and Rahul Savani} 

\ccsdesc{Theory of computation~Design and analysis of algorithms}

\keywords{query complexity, Tarski fixed points, total function problem} 

\category{} 

\supplement{}

\acknowledgements{We would like to thank Kousha Etessami, Thomas
Webster, and an anonymous reviewer for pointing out that the proof of Lemma 12
could be drastically simplified from its original version, and we would like to
thank Bal\'azs Keszegh for useful discussions on this topic.}

\nolinenumbers 

\hideLIPIcs  

\EventEditors{John Q. Open and Joan R. Access}
\EventNoEds{2}
\EventLongTitle{42nd Conference on Very Important Topics (CVIT 2016)}
\EventShortTitle{STACS 2021}
\EventAcronym{CVIT}
\EventYear{2016}
\EventDate{December 24--27, 2016}
\EventLocation{Little Whinging, United Kingdom}
\EventLogo{}
\SeriesVolume{42}
\ArticleNo{23}

\newcommand{\midp}{\ensuremath{\textsf{mid}}\xspace}
\newcommand{\midpone}{\ensuremath{\textsf{midone}}\xspace}
\newcommand{\midptwo}{\ensuremath{\textsf{midtwo}}\xspace}
\newcommand{\topp}{\ensuremath{\textsf{top}}\xspace}
\newcommand{\botp}{\ensuremath{\textsf{bot}}\xspace}
\newcommand{\leftp}{\ensuremath{\textsf{left}}\xspace}
\newcommand{\rightp}{\ensuremath{\textsf{right}}\xspace}

\usepackage{amsmath}

\usepackage{relsize}

\usepackage{pgfplots}
\usepackage{pgfplotstable}
\usetikzlibrary{matrix,arrows,shapes,positioning,calc,snakes,decorations.markings}
\usetikzlibrary{fadings,shapes.arrows,shadows}   
\usetikzlibrary{arrows.meta,patterns}
\usepackage{tkz-euclide}

\usetikzlibrary{arrows.meta}
\tikzset{>={Latex[width=3mm,length=2.25mm]}}

\pgfdeclarelayer{background}
\pgfsetlayers{background,main}

\definecolor{lightgray}{gray}{0.9}
\definecolor{medgray}{gray}{0.55}
\definecolor{darkgray}{gray}{0.4}

\definecolor{pastelblue}{rgb}{0.4, 0.5, 0.91}
\definecolor{pastelred}{rgb}{0.9, 0.18, 0.10}

\definecolor{darkpastelgreen}{rgb}{0.01, 0.75, 0.24}
\definecolor{pastelyellow}{rgb}{0.01, 0.75, 0.24}

\definecolor{shade1}{rgb}{0.87, 0.91, 0.92} 
\definecolor{shade2}{rgb}{1, 1, 0.88} 
\definecolor{shade3}{rgb}{0.69, 0.94, 0.92} 
\definecolor{shade4}{rgb}{1, 0.79, 0.78} 

\definecolor{darkpastelblue}{rgb}{0.27, 0.47, 0.70}
\definecolor{darkpastelred}{rgb}{0.59, 0.18, 0.10}
\definecolor{darkpastelbrown}{rgb}{0.39, 0.31, 0.25}
\definecolor{darkpastelmagenta}{rgb}{0.59, 0.44, 0.84}

\tikzset{fontscale/.style = {font=\relsize{#1}}}				
\tikzset{query/.style={draw, circle, fill=black, inner sep=2pt, minimum size=0.05cm}}

\pgfmathsetmacro{\diag}{0.2 + sqrt(1/2)}

\pgfmathsetmacro{\myscale}{0.65}

\pgfmathsetmacro{\lw}{0.9mm}
\pgfmathsetmacro{\ll}{1.2}
\pgfmathsetmacro{\lls}{0.5}

\newcommand{\myuparrow}[2]{
	\draw[->, line width=\lw, color=pastelblue] (#1,#2) -- (#1,#2+\ll);
}

\newcommand{\mydnarrow}[2]{
	\draw[->, line width=\lw, color=pastelblue] (#1,#2) -- (#1,#2-\ll);
}

\newcommand{\myuparrowshort}[2]{
	\draw[->, line width=\lw, color=pastelblue] (#1,#2) -- (#1,#2+\lls);
}

\newcommand{\mydnarrowshort}[2]{
	\draw[->, line width=\lw, color=pastelblue] (#1,#2) -- (#1,#2-\lls);
}

\newcommand{\myleftarrow}[2]{
	\draw[->, line width=\lw, color=pastelred] (#1,#2) -- (#1-\ll,#2);
}

\newcommand{\myrightarrow}[2]{
	\draw[->, line width=\lw, color=pastelred] (#1,#2) -- (#1+\ll,#2);
}

\newcommand{\myleftarrowshort}[2]{
	\draw[->, line width=\lw, color=pastelred] (#1,#2) -- (#1-\lls,#2);
}

\newcommand{\myrightarrowshort}[2]{
	\draw[->, line width=\lw, color=pastelred] (#1,#2) -- (#1+\lls,#2);
}

\newcommand{\backarrow}[3]{
	\draw[->, line width=\lw, color=#3] (#1,#2) -- (#1+\diag,#2+\diag);
}

\newcommand{\forwardarrow}[3]{
	\draw[->, line width=\lw, color=#3] (#1,#2) -- (#1-\diag, #2-\diag);
}

\newcommand{\subfig}[4]{
	\node[label={[yshift=2mm]above:{\Huge #3}}] at (#1,#2) {\usebox#4};
}
\newcommand{\subfigprime}[5]{
	
\node[label={[yshift=2mm]above:{{\renewcommand*{\arraystretch}{2.5} \begin{tabular}{c}\Huge #3\\ \Huge #4 \end{tabular}}}}] at (#1,#2) {\usebox#5};
}

\usepackage{subcaption}

\newcommand{\myparagraph}[1]{\smallskip \noindent \textbf{#1}}

\usepackage{enumitem}
\usepackage{xspace}
\usepackage{todonotes}

\newcommand{\nats}{\mathbb N}

\newcommand{\solnlabel}[1]{\label{sol:#1}}
\newcommand{\solnref}[1]{\ref{sol:#1}}
\newcommand{\blank}{\ensuremath{\mathtt{*}}}

\def\problem#1{{\scshape #1}}
\def\tarski{\problem{Tarski}\xspace}

\DeclareMathOperator{\poly}{poly}

\DeclareMathOperator{\Up}{Up}
\DeclareMathOperator{\Down}{Down}

\DeclareMathOperator{\lat}{Lat}

\begin{document}

\maketitle

\begin{abstract}
Dang et al.\ have given an algorithm that can find a Tarski fixed point in a
$k$-dimensional lattice of width $n$ using $O(\log^{k} n)$ queries~\cite{DQY20}.
Multiple authors have conjectured that this algorithm is optimal~\cite{DQY20,
EPRY20}, and indeed this has been proven for two-dimensional
instances~\cite{EPRY20}. We show that these conjectures are false in dimension
three or higher by giving an $O(\log^2 n)$ query algorithm for the
three-dimensional Tarski problem. We also give a new decomposition theorem for
$k$-dimensional Tarski problems which, in combination with our new algorithm for three
dimensions, gives an $O(\log^{2 \lceil
k/3 \rceil} n)$ query algorithm for the $k$-dimensional problem. 
\end{abstract}

\newpage

\section{Introduction}

Tarski's fixed point theorem states that every order preserving function on a
complete lattice has a greatest and least fixed point~\cite{Tarski55}, and
therefore in particular, every such function has at least one fixed point.
Recently, there has been interest in the complexity of finding
such a fixed point. This is due to its applications, including computing Nash
equilibria of supermodular games and finding the solution of a
simple stochastic game~\cite{EPRY20}. 

Prior work has focused on the complete lattice $L$ defined by a $k$-dimensional grid
of width $n$. Dang, Qi, and Ye~\cite{DQY20} give an algorithm that finds a fixed point of
a function $f : L \rightarrow L$ using $O(\log^k n)$ queries to
$f$.  This algorithm uses recursive binary search, where a $k$-dimensional
problem is solved by making $\log n$ recursive calls on $(k-1)$-dimensional
sub-instances. They conjectured that this algorithm is optimal.

Later work of Etessami, Papadimitriou, Rubinstein, and Yannakakis took the first
step towards proving this~\cite{EPRY20}. They showed that finding a Tarski fixed point in a
two-dimensional lattice requires $\Omega(\log^2 n)$ queries, meaning that the Dang et
al.\ algorithm is indeed optimal in the two-dimensional case. Etessami et al.\
conjectured that the Dang et al. algorithm is optimal for constant $k$, and they
leave as an explicit open problem the question of whether their lower bound can
be extended to dimension three or beyond. 

\myparagraph{Our contribution.}
In this paper we show that, surprisingly, the Dang et al. algorithm is not
optimal in dimension three, or any higher dimension, and so we falsify
the prior conjectures. We do this by giving an algorithm that can find a Tarski
fixed point in three dimensions using $O(\log^2 n)$ queries, thereby
beating the $O(\log^3 n)$ query algorithm of Dang et al. 


The Dang et al.\ algorithm solves a three-dimensional instance by making
recursive calls to find a fixed point of $\log n$ distinct two-dimensional
sub-instances. Our key innovation is to point out that one does not need to find
a fixed point of the two-dimensional sub-instance to make progress. Instead, we
define the concept of an \emph{inner algorithm} (Definition~\ref{def:inner})
that, given a two-dimensional sub-instance, is permitted to return any point
that lies in the up or down set of the three-dimensional instance (defined
formally later). This is a much larger set of points, so whereas finding a fixed
point of a two-dimensional instance requires $\Omega(\log^2 n)$
queries~\cite{EPRY20}, we give a $O(\log n)$ query inner algorithm for
two-dimensional instances. This inner algorithm is quite involved, and is the
main technical contribution of the paper.

We show that, given an inner algorithm for dimension $k-1$, a reasonably
straightforward \emph{outer algorithm} can find a Tarski fixed point by making
$O(k \cdot \log n)$ calls to the inner algorithm. Thus we obtain a $O(\log^2 n)$
query algorithm for the case where $k=3$. We leave as an open problem the
question of whether efficient inner algorithms exist in higher dimensions.

For higher-dimensional instances, we show a decomposition theorem: if
$a$-dimensional Tarski problems can be solved in $q_a$ queries, and
$b$-dimensional Tarski problems can be solved in $q_b$ queries, then $(a\cdot
b)$-dimensional \tarski can be solved in $q_a \cdot (q_b + 2)$ queries. This
then allows us to use our new algorithm for three-dimensional Tarski problems to
obtain an algorithm that solves $k$-dimensional Tarski problems using $O(\log^{2
\lceil k/3 \rceil} n)$ queries, a substantial improvement over the $O(\log^k n)$
algorithm of Dang et al.\cite{DQY20}.

Though we state our results in terms of query complexity for the sake of
simplicity, it should be pointed out that all of our algorithms
run in polynomial time. 
Specifically, our algorithms will run in $O(\poly(\log n, k) \cdot
\log^{2 \lceil k/3 \rceil}n)$ time when the function is presented as a Boolean circuit of size 
$\poly(\log n, k)$.

\myparagraph{Related work.}
Etessami et al.\ also studied the computational complexity of the Tarski
problem~\cite{EPRY20}, showing that the problem lies in 
PPAD and PLS. However, the exact complexity of the problem remains
open. It is not clear whether the problem is
PPAD~$\cap$~PLS-complete~\cite{FGHS20}, or contained in some other lower class such as
EOPL or UEOPL~\cite{FGMS20}.

Tarski's fixed point theorem has been applied in a wide range of settings within
Economics~\cite{Topkis79,MilgromR90,Topkis98}, and in particular to settings
that can be captured by supermodular games, which are in fact equivalent to the
Tarski problem~\cite{EPRY20}. In terms of algorithms, Echenique
\cite{Echenique07} studied the problem of computing all pure equilibria of a
supermodular game, which is at least as hard as finding the greatest or least
fixed point of the Tarski problem, which is itself NP-hard~\cite{EPRY20}.
There have also been several papers that study properties of Tarski fixed
points, such as the complexity of deciding whether a fixed point is
unique~\cite{DQY20,DangY18,DangYe18-tech,DangY20}. The Tarski problem has also
been studied in the setting where the partial order is given by an oracle~\cite{ChangLT08}.

\section{Preliminaries}

\myparagraph{\bf Lattices.}
We work with a complete lattice defined over a
$k$-dimensional grid of points. We define $\lat(n_1, n_2, \dots, n_k)$ to be the
$k$-dimensional lattice with side-lengths given by $n_1, \ldots, n_k$. That is,
$\lat(n_1, n_2, \dots, n_k)$ contains every $x \in \nats^k$ such that $1 \le x_i
\le n_i$ for all $i=1,\ldots,k$.
Throughout, we use $k$ to denote the dimensionality of the
lattice, and $n = \max_{i = 1}^k n_i$ to be the width of the widest dimension.
We use $\preceq$ to denote the natural partial order over this lattice with
$x \preceq y$ if and only if $x, y \in L$ and $x_i \le y_i$ for all $i \le
k$. 

\myparagraph{\bf The Tarski fixed point problem.}
Given a lattice $L$, a function $f : L \rightarrow L$ is \emph{order preserving} if $f(x) \preceq
f(y)$ whenever $x \preceq y$. A point $x \in L$ is \emph{fixed point} of $f$ if
$f(x) = x$. A weak version of Tarski's theorem can be stated as follows. 
\begin{theorem}[\cite{Tarski55}]
Every order preserving function on a complete lattice has a fixed point.
\end{theorem}
Thus, we can define a total search problem for Tarski's fixed point theorem.
\begin{definition}[\tarski]
Given a lattice $L$, and a function $f : L \rightarrow L$, find one of:
\begin{enumerate}[label=(T\arabic*), wide=0pt]
\item \solnlabel{T1} A point $x \in L$ such that $f(x) = x$.
\item \solnlabel{T2} Two points $x, y \in L$ such that $x \preceq y$ and $f(x) \not\preceq
f(y)$.
\end{enumerate}
\end{definition}
Solutions of type \solnref{T1} are fixed points of $f$, whereas
solutions of type \solnref{T2} witness that~$f$ is not an order
preserving function. 
By Tarski's theorem, if a function $f$ has no
solutions of type \solnref{T2}, then it must have a solution of type
\solnref{T1}, and so \tarski is a total problem. 

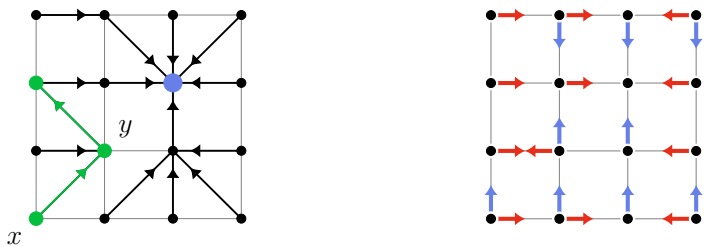
\begin{figure}
\begin{center}
\scalebox{0.6}{\newsavebox\exleft
\begin{lrbox}{\exleft}
\begin{tikzpicture}[scale=1.5]

\tikzset{every node}=[font=\LARGE]

\draw[gray] (0, 0) grid (3, 3);

\begin{scope}[very thick,decoration={
    markings,
    mark=at position 0.75 with {\arrow{>}}}
    ] 
	\draw[postaction={decorate}] (0,0) -- (1,1);
	\draw[postaction={decorate}] (0,1) -- (1,1);
	\draw[postaction={decorate}] (0,2) -- (1,2);
	\draw[postaction={decorate}] (0,3) -- (1,3);
	\draw[postaction={decorate}] (1,0) -- (2,1);
	\draw[postaction={decorate}] (1,1) -- (0,2);
	\draw[postaction={decorate}] (1,2) -- (2,2);
	\draw[postaction={decorate}] (1,3) -- (2,2);
	\draw[postaction={decorate}] (2,0) -- (2,1);
	\draw[postaction={decorate}] (2,1) -- (2,2);
	\draw[postaction={decorate}] (2,3) -- (2,2);
	\draw[postaction={decorate}] (3,0) -- (2,1);
	\draw[postaction={decorate}] (3,1) -- (2,1);
	\draw[postaction={decorate}] (3,2) -- (2,2);
	\draw[postaction={decorate}] (3,3) -- (2,2);

	\draw[postaction={decorate},darkpastelgreen] (0,0) -- (1,1);
	\draw[postaction={decorate},darkpastelgreen] (1,1) -- (0,2);
\end{scope}

\foreach \x in {0,...,3}
\foreach \y in {0,...,3}
	\node[query] at (\x,\y) {};

\node[query,darkpastelgreen,minimum size=3mm,label=below left:$x$] at (0,0) {};
\node[query,darkpastelgreen,minimum size=3mm,label=above right:$y$] at (1,1) {};
\node[query,darkpastelgreen,minimum size=3mm] at (0,2) {};
\node[query,pastelblue,minimum size=4mm] at (2,2) {};

\end{tikzpicture}
\end{lrbox}

\newsavebox\exright
\begin{lrbox}{\exright}
\begin{tikzpicture}[scale=1.5]

\tikzset{every node}=[font=\LARGE]

\draw[gray] (0, 0) grid (3, 3);

\myuparrowshort{0}{0}

\myrightarrowshort{0}{0}
\myrightarrowshort{0}{1}
\myrightarrowshort{0}{2}
\myrightarrowshort{0}{3}

\myuparrowshort{1}{0}
\myuparrowshort{1}{1}
\mydnarrowshort{1}{3}

\myrightarrowshort{1}{0}
\myleftarrowshort{1}{1}
\myrightarrowshort{1}{2}
\myrightarrowshort{1}{3}

\myuparrowshort{2}{0}
\myuparrowshort{2}{1}
\mydnarrowshort{2}{3}

\myuparrowshort{3}{0}
\mydnarrowshort{3}{3}

\myleftarrowshort{3}{0}
\myleftarrowshort{3}{1}
\myleftarrowshort{3}{2}
\myleftarrowshort{3}{3}

\node[query,white,minimum size=3mm,label={[white] below left:$x$}] at (0,0) {};

\foreach \x in {0,...,3}
\foreach \y in {0,...,3} { 
	\node[query,white,minimum size=3mm] at (\x,\y) {};
	\node[query] at (\x,\y) {};
}

\end{tikzpicture}
\end{lrbox}

\begin{tikzpicture}

\subfig{0}{0}{}{\exleft}
\subfig{10}{0}{}{\exright}

\end{tikzpicture}
}
\end{center}
\caption{Left: a \tarski instance. Right: our diagramming notation for
the same instance.}
\label{fig:2dex}
\end{figure}

The left-hand picture in Figure~\ref{fig:2dex} gives an example of a
two-dimensional \tarski instance. The blue point is a fixed point, and so is a
\solnref{T1} solution, while the highlighted green arrows give an
example of an order preservation violation, and so $(x,y)$ is 
a \solnref{T2} solution.

Throughout the paper we will use a diagramming notation, shown on the right in
Figure~\ref{fig:2dex}, that decomposes
the dimensions of the instance. The red arrows correspond to dimension 1, where an
arrow pointing to the left indicates that $f(x)_1 \le x_1$, while an arrow to
the right\footnote{If $x_1 = f(x)_1$ we could use either arrow, but
will clarify in the text whenever this ambiguity matters.} indicates
that $x_1 \le f(x)_1$. Blue arrows do the same thing for dimension $2$, and we
will use green arrows for dimension 3 in the cases where this is relevant.  


\myparagraph{\bf The up and down sets.}
Given a function $f$ over a lattice
$L$, we define 
$\Up(f) = \{ x \in L \; : \; x \preceq f(x) \}$, and
$\Down(f) = \{ x \in L \; : \; f(x) \preceq x \}$.
We call $\Up(f)$, the \emph{up set}, which contains all points in which $f$ goes
up according to the ordering~$\preceq$, and likewise we call $\Down(f)$ the
\emph{down set}. Note that the set of fixed points of $f$ is exactly $\Up(f) \cap
\Down(f)$. 

\myparagraph{\bf Slices.}
A \emph{slice} of the lattice $L$ is defined by a tuple $s = (s_1, s_2, \dots,
s_k)$, where each $s_i \in \mathbb{N} \cup \{ \blank \}$. The idea is that, if
$s_i \ne \blank$, then we fix dimension $i$ of $L$ to be $s_i$, and if $s_i =
\blank$, then we allow dimension $i$ of $L$ to be free. Formally, we define the
sliced lattice 
$L_s = \{ x \in L \; : \; x_i = s_i \text{ whenever } s_i \ne \blank
\}.$
%
We say that a slice is a \emph{principle slice} if it fixes exactly one
dimension and leaves the others free. For example $(1, \blank, \blank)$,
$(\blank, 33, \blank)$, and $(\blank, \blank, 261)$ are all principle slices of a
three-dimensional lattice. 

Given a slice $s$, and a function $f : L \rightarrow L$, we define $f_s : L_s
\rightarrow L_s$ to be the \emph{restriction} of~$f$ to $L_s$. Specifically, for
each $x \in L_s$, we define 
$\left(f_s(x)\right)_i = f(x)_i$ if $s_i = \blank$, and 
$\left(f_s(x)\right)_i = s_i$ otherwise.
This definition projects the function $f$ down onto the slice $s$.

A fact that we will use repeatedly in the paper is that
an order preservation violation in a slice~$s$ is also an order preservation
violation for the whole instance. More formally, if $x, y \in L_s$
satisfy $x \preceq y$ and $f_s(x) \not\preceq f_s(y)$, then we also have $f(x)
\not \preceq f(y)$, since there exists a dimension $i$ such that $f(x)_i =
f_s(x)_i > f_s(y)_i = f(y)_i$.

\myparagraph{\bf Sub-instances.} 
A \emph{sub-instance} of a lattice $L$ is defined by two points $x, y \in L$
that satisfy $x \preceq y$. Informally, the sub-instance defined by $x$ and $y$ is
the lattice containing all points between $x$ and $y$. Formally, we define
$L_{x, y} = \{ a \in L \; : \; x \preceq a \preceq y\}$.




\section{The Outer Algorithm}
\label{sec:outer}

The task of the outer algorithm is to find a solution to the \tarski instance
by making $O(k \cdot \log n)$ calls to the inner algorithm. 
We state our results for dimension $k$, even though we only apply the outer
algorithm with $k=3$, since, in the future, an efficient inner algorithm in
higher dimensions may be found. Formally, an inner algorithm is defined as
follows.

\begin{definition}[Inner algorithm]
\label{def:inner}
An inner algorithm for a \tarski instance takes as input a sub-instance $L_{a,
b}$ with $a \in \Up(f)$ and $b \in \Down(f)$, and a principle slice $s$ of that
sub-instance. It outputs one of the following.
\begin{itemize}
\item A point $x \in L_{a, b} \cap L_s$ such that $x \in \Up(f)$ or $x \in \Down(f)$. 
\item Two points $x, y \in L_{a, b}$ that witness a violation of the order
preservation of $f$.
\end{itemize}
\end{definition}
It is important to understand that here we are looking for points that lie in
the up or down set of the \emph{three-dimensional} instance, a point that lies
in $\Up(f_s)$ but for which $f$ goes down in the third dimension would not
satisfy this criterion.




\myparagraph{\bf The algorithm.}
Throughout the outer algorithm, we will maintain two points $x, y \in L$ with
the invariant that $x \preceq y$ and $x \in \Up(f)$ and $y \in \Down(f)$.  
The following lemma, which is proved in Appendix~\ref{app:tarskinew}, implies
that if $x$ and $y$ satisfy the invariant, then $L_{x, y}$ must contain a
solution to the \tarski problem. This will allow us to focus on smaller and
smaller instances that are guaranteed to contain a solution.

\begin{lemma}
\label{lem:tarskinew}
Let $L$ be a lattice and $f : L \rightarrow L$ be a \tarski instance. If there
are two points $a, b \in L$ satisfying $a \preceq b$, $a \in \Up(f)$, and $b \in
\Down(f)$, then one of the following exists.
\begin{itemize}
\item A point $x \in L_{a, b}$ satisfying $f(x) = x$.
\item Two points $x, y \in L_{a, b}$ satisfying $x \preceq y$ and $f(x)
\not\preceq f(y)$.
\end{itemize}
Moreover, there is an algorithm that finds one of the above using
$O(\sum_{i = 1}^k (a_i - b_i))$ queries.
\end{lemma}

Initially we set $x = (1, 1, \dots, 1)$, which is the least element, and $y = (n_1, n_2, \dots, n_k)$, which is the greatest element. 
Note that $x \preceq
f(x)$ holds because $x$ is the least element, and likewise $f(y)
\preceq y$ holds because $y$ is the greatest element, so the invariant holds for
these two points.

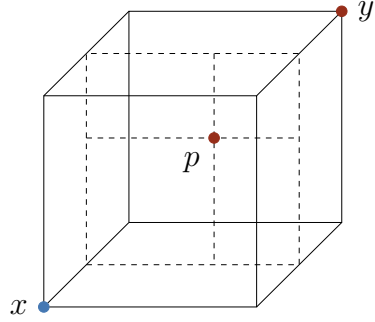
\begin{figure}
\begin{center}
\scalebox{0.7}{\begin{tikzpicture}[scale=0.8]

\tikzset{every node}=[font=\LARGE]

\pgfmathsetmacro{\side}{5}
\pgfmathsetmacro{\offset}{2}

\node[query,darkpastelblue,label=left:$x$] (a) at (0,0) {};

\coordinate (b) at (\side,0);
\coordinate (c) at (0,\side);
\coordinate (d) at (\side,\side);

\coordinate (aa) at ($(a) + (\offset,\offset)$);
\coordinate (bb) at ($(b) + (\offset,\offset)$);
\coordinate (cc) at ($(c) + (\offset,\offset)$);
\node[query,darkpastelred,label=right:$y$] (dd) at ($(d) + (\offset,\offset)$) {};

\coordinate (aaa) at ($(a)!0.5!(aa)$);
\coordinate (bbb) at ($(b)!0.5!(bb)$);
\coordinate (ccc) at ($(c)!0.5!(cc)$);
\coordinate (ddd) at ($(d)!0.5!(dd)$);


\coordinate (l) at ($(ccc)!0.4!(aaa)$);
\coordinate (r) at ($(bbb)!0.6!(ddd)$);
\coordinate (top) at ($(ccc)!0.6!(ddd)$);
\coordinate (bot) at ($(aaa)!0.6!(bbb)$);

\draw[dashed] (l) -- (r);
\draw[dashed] (top) -- (bot);

\tkzInterLL(l,r)(top,bot) \tkzGetPoint{newred}

\node[query,darkpastelred,label=below left:$p$] at (newred) {};

\draw (a) -- (b) -- (d) -- (c) -- (a);
\draw (aa) -- (bb) -- (dd) -- (cc) -- (aa);
\draw[dashed] (aaa) -- (bbb) -- (ddd) -- (ccc) -- (aaa);

\draw (a) -- (aa);
\draw (b) -- (bb);
\draw (c) -- (cc);
\draw (d) -- (dd);

\end{tikzpicture}
}
\end{center}
\caption{One iteration of the outer algorithm. The dashed lines show the
principle slice chosen by the algorithm, and the point $p$ is the point returned
by the inner algorithm. In this case $p \in \Down(f)$, and so the algorithm
focuses on the sub-instance $L_{x, p}$.}
\label{fig:outer}
\end{figure}

Each iteration of the outer algorithm will reduce the number of points in $L_{x,
y}$ by a factor of two.
To do this, the algorithm selects a largest dimension of
that sub-instance, which is a dimension $i$ that maximizes $y_i - x_i$. It
then makes a call to the inner algorithm for the principle slice $s$ defined so that $s_i = \lfloor (y_i - x_i)/2 \rfloor$ and $s_j = \blank$ for all $j \ne i$.
\begin{enumerate}
\item If the inner algorithm returns a violation of order preservation in the
slice, then this is also an order preservation violation in $L$, and so the
algorithm returns this and terminates. 
\item If the inner algorithm returns a point $p$ in the slice such that $p \in \Up(f)$, then
the algorithm sets $x := p$ and moves to the next iteration.
\item If the inner algorithm returns a point $p$ such that $p \in \Down(f)$, then
the algorithm sets $y := p$ and moves to the next iteration.
\end{enumerate}
Figure~\ref{fig:outer} gives an example of this procedure.

The algorithm can continue as long as there exists a dimension $i$ such that
$y_i - x_i \ge 2$, since this ensures that there will exist a principle slice
strictly between $x$ and $y$ in dimension~$i$ that cuts the sub-instance
in half. Note that there can be at most $k \cdot \log n$ iterations of the
algorithm before we arrive at the final sub-instance $L_{x, y}$
with $y_i - x_i < 2$ for all $i$. Lemma~\ref{lem:tarskinew} gives us an
efficient algorithm to find a solution in this final instance, which uses at
most $O(\sum_{i = 1}^k (y_i - x_i)) = O(k)$ queries.
So we have proved the following theorem.


\begin{theorem}
\label{thm:outer}
Suppose that there exists an inner algorithm that makes at most $q$ queries.
Then a solution to the \tarski problem can be found by making $O(q \cdot k \cdot \log n + k)$ queries. 
\end{theorem}

\section{The Inner Algorithm}
\label{sec:inner}

We now describe an inner algorithm for three dimensions that makes
$O(\log n)$ queries. Throughout this section we assume that the inner
algorithm has been invoked on a sub-instance $L_{u, d}$ and
principle slice $s$, and without loss of
generality we assume that $s = (\blank, \blank, s_3)$.

\begin{figure}
\begin{center}
\scalebox{0.4}{

\newsavebox\one
\begin{lrbox}{\one}
\begin{tikzpicture}[scale=\myscale]
\tikzset{every node}=[font=\Huge]
\draw (0, 0) rectangle (8, 8);
\mydnarrow{8}{8}
\myleftarrow{8}{8}
\myuparrow{0}{0}
\myrightarrow{0}{0}
\node[query, label=left:$a$] at (0,0) {};
\node[query, label=right:$b$] at (8,8) {};
\end{tikzpicture}
\end{lrbox}

\newsavebox\two
\begin{lrbox}{\two}
\begin{tikzpicture}[scale=\myscale]
\tikzset{every node}=[font=\Huge]
\draw (0, 0) rectangle (8, 8);
\forwardarrow{0}{0}{darkpastelgreen}
\forwardarrow{5}{0}{darkpastelgreen}
\myrightarrow{0}{0}
\myleftarrow{5}{0}
\node[query,label={[yshift=2mm] left:$a$}] (a) at (0,0) {};
\node[query,label=above:$u$] (a) at (5,0) {};
\backarrow{8}{8}{darkpastelgreen}
\backarrow{5}{8}{darkpastelgreen}
\myleftarrow{8}{8}
\myrightarrow{5}{8}
\node[query,label={[yshift=-2mm] right:$b$}] (c) at (8,8) {};
\node[query,label=below:$d$] (d) at (5,8) {};
\end{tikzpicture}
\end{lrbox}

\newsavebox\three
\begin{lrbox}{\three}
\begin{tikzpicture}[scale=\myscale]
\tikzset{every node}=[font=\Huge]
\draw (0, 0) rectangle (8, 8);
\forwardarrow{0}{0}{darkpastelgreen}
\forwardarrow{0}{5}{darkpastelgreen}
\myuparrow{0}{0}
\mydnarrow{0}{5}
\node[query,label={[yshift=2mm] left:$a$}] (a) at (0,0) {};
\node[query,label=right:$u$] (a) at (0,5) {};
\backarrow{8}{8}{darkpastelgreen}
\backarrow{5}{8}{darkpastelgreen}
\myleftarrow{8}{8}
\myrightarrow{5}{8}
\node[query,label={[yshift=-2mm] right:$b$}] (c) at (8,8) {};
\node[query,label=below:$d$] (d) at (5,8) {};
\end{tikzpicture}
\end{lrbox}

\newsavebox\four
\begin{lrbox}{\four}
\begin{tikzpicture}[scale=\myscale]
\tikzset{every node}=[font=\Huge]
\draw (0, 0) rectangle (8, 8);
\myuparrow{0}{0}
\myrightarrow{0}{0}
\node[query, label=left:$a$] at (0,0) {};
\backarrow{8}{8}{darkpastelgreen}
\backarrow{8}{5}{darkpastelgreen}
\mydnarrow{8}{8}
\myuparrow{8}{5}
\node[query,label={[yshift=-2mm] right:$b$}] (c) at (8,8) {};
\node[query,label=left:$d$] (d) at (8,5) {};
\end{tikzpicture}
\end{lrbox}


\begin{tikzpicture}

\subfig{0}{0}{}{\one}
\subfig{9}{0}{}{\two}
\subfig{18}{0}{}{\three}
\subfig{27}{0}{}{\four}

\end{tikzpicture}
}
\end{center}
\caption{Four example sub-instances that satisfy the inner algorithm invariant. 
}
\label{fig:invariant}
\end{figure}
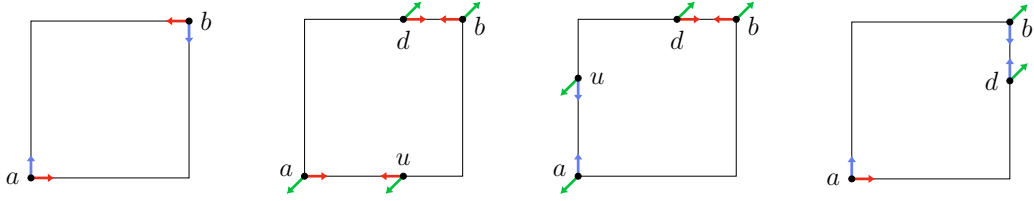

\myparagraph{\bf Down set witnesses.}
Like the outer algorithm, the inner algorithm will also focus on smaller and
smaller sub-instances that are guaranteed to contain a solution by an invariant,
but now the invariant is more complex.
%
To define the invariant, we first introduce the concept of a \emph{down
set witness} and an \emph{up set witness}. The points $d$ and $b$ in the second
example in Figure~\ref{fig:invariant} give an example of a down set witness.
Note that the following properties are satisfied. 
\begin{itemize}
\item $f$ weakly increases at $d$ and $b$ in dimension 3.
\item $d$ and $b$ have the same coordinate in dimension 2.
\item $d$ weakly increases in dimension 1 while $b$ weakly decreases in dimension 1.
\end{itemize}
We also
allow down set witnesses like those given by $d$ and $b$ in the fourth example
of Figure~\ref{fig:invariant} that satisfy the same properties with
dimensions 1 and 2 swapped. Thus, the 
formal definition of a down set witness abstracts over dimensions 1 and 2.

\begin{definition}[Down set witness]
A down set witness is a pair of points $(d, b)$ with $d, b \in L_s$ 
such that both of the following are satisfied. 
\begin{itemize}
\item $d_3 \le f(d)_3$ and $b_3 \le f(b)_3$.
\item $\exists\ i,j \in \{1, 2\}$ with $i \ne
j$ s.t. $d_i = b_i$ and $d_j \le b_j$, while $d_j \le f(d)_j$ and $f(b)_j \le b_j$.
\end{itemize}
\end{definition}

\noindent If $(d, b)$ is a down set witness and $d_2 = b_2$, then we call $(d,
b)$ a \emph{top-boundary} witness, while if $d_1 = b_1$, then we call $(d, b)$ a
\emph{right-boundary} witness.

The following lemma states that if we have a down set witness $(d, b)$, then
between~$d$ and~$b$ we can find
either a solution that can be returned by the inner algorithm (cases 1
and~2 of the lemma), or a point that is in the down set of the slice $s$ (case
3 of the lemma).

The proof the lemma can be found in Appendix~\ref{app:downsetwitness}.
Informally, the proof for a top-boundary witness $(d, b)$
uses the fact that $d$ and $b$ point towards each other in dimension $1$ to
argue that there must be a fixed point $p$ (or an order preservation violation)
of the one-dimensional slice between $d$ and $b$. Then, it is shown that either
this point is in $\Up(f)$, and so is a solution for the inner algorithm, or it
is in $\Down(f_s)$, or that $p$ violates order preservation with $d$ or
$b$. 

\begin{lemma}
\label{lem:downsetwitness}
If $(d, b)$ is a down set witness, then one of the following exists.
\begin{enumerate}
\item A point $c$ satisfying $d \preceq c \preceq b$ such that $c \in \Up(f)$.
\item Two points $x, y$ satisfying $d \preceq x \preceq y \preceq b$ that
witness order preservation violation of~$f$.
\item A point $c$ satisfying $d \preceq c \preceq b$ such that $c \in
\Down(f_s)$.
\end{enumerate}
\end{lemma}

\myparagraph{\bf Up set witnesses.}
An up set witness is simply a down set witness in which all inequalities have
been flipped. The second and third diagrams in Figure~\ref{fig:invariant} show
the two possible configurations of an up set witness $(a, u)$. Note that for up
set witnesses, dimension 3 is now required to weakly decrease.

\begin{definition}[Up set witness]
An up set witness is a pair of points $(a, u)$ with $a, u \in L_s$ such that
both of the following are satisfied. 
\begin{itemize}
\item $a_3 \ge f(a)_3$ and $u_3 \ge f(u)_3$.
\item $\exists\ i,j \in \{1, 2\}$ with $i \ne j$ s.t. $a_i = u_i$ and $u_j \ge a_j$, while
$u_j \ge f(u)_j$ and $f(a)_j \ge a_j$.
\end{itemize}
\end{definition}

\noindent We say that an up set witness $(a, b)$ is a \emph{left-boundary}
witness if $a_1 = b_1$, while we call it a \emph{bottom-boundary} witness if
$a_2 = b_2$.

The following lemma is the analogue of Lemma~\ref{lem:downsetwitness} for up set
witnesses. 
The proof, which can be found in Appendix~\ref{app:upsetwitness}, simply flips all inequalities in the proof of Lemma~\ref{lem:downsetwitness}.

\begin{lemma}
\label{lem:upsetwitness}
If $(a, u)$ is an up set witness, then one of the following exists.
\begin{enumerate}
\item A point $c$ satisfying $a \preceq c \preceq u$ such that $c \in \Down(f)$.
\item Two points $x, y$ satisfying $a \preceq x \preceq y \preceq u$ that
witness order preservation violation of~$f$.
\item A point $c$ satisfying $a \preceq c \preceq u$ such that $c \in \Up(f_s)$.
\end{enumerate}
\end{lemma}

\myparagraph{\bf The invariant.}
At each step of the inner algorithm, we will have a sub-instance $L_{a, b}$ 
that satisfies the following invariant. 

\begin{definition}[Inner algorithm invariant]
The instance $L_{a, b}$ satisfies the invariant if 
\begin{itemize}
\item Either $a \in \Up(f_s)$ or there is a known up set witness $(a, u)$ with
$u \preceq b$.
\item Either $b \in \Down(f_s)$ or there is a known down set witness $(d, b)$
with $a \preceq d$.
\end{itemize}
If we have both an up set witness and a down set witness then we also
require that $u \preceq d$.
\end{definition}

Figure~\ref{fig:invariant} gives four example instances that satisfy the
invariant. Note that there are actually nine possible configurations, since the
first point of the invariant can be satisfied either by a point in the up set, a
left-boundary up set witness, or a bottom-boundary up set witness, and the
second point of the invariant likewise has three possible configurations.

The following lemma shows that, if the invariant is satisfied, then the
sub-instance $L_{a, b}$ contains a solution that can be returned by the inner
algorithm. The proof invokes Lemmas~\ref{lem:downsetwitness}
and~\ref{lem:upsetwitness} to either immediately find a solution for the inner
algorithm, or find two points $x \preceq y$ in the sub instance where $x$ is in
the up set and $y$ is in the down set. The latter case allows us to invoke
Lemma~\ref{lem:tarskinew} to argue that the sub-instance contains a fixed point
$p$ of the slice~$s$. If $p$ weakly increases in the third dimension, then $p
\in \Up(f)$, while if $p$ decreases in the third dimension then $p \in
\Down(f)$.
The full proof can be found in Appendix~\ref{app:invariant}.

\begin{lemma}
\label{lem:invariant}
If $L_{a, b}$ satisfies the invariant then one of the following exists.
\begin{itemize}
\item A point $x \in L_{a, b}$ such that $x \in \Up(f)$ or $x \in \Down(f)$. 
\item Two points $x, y \in L_{a, b}$ that witness a violation of the order preservation of $f$.
\end{itemize}
\end{lemma}


\myparagraph{\bf A special case.}
There is a special case that we will encounter in the inner algorithm that
requires more effort to deal with. One example of this case is shown in
Case 3.a.ii of Figure~\ref{fig:step2}. Here we have a point $p$ on the right-hand boundary of
the instance that satisfies $p_1 < f(p)_1$, meaning that $f$ moves $p$ outside
of the instance. If $b \in \Down(f)$, or if there is a top-boundary down set
witness, then it is straightforward to show that $p$ and $b$ violate order
preservation. 

However, if we have a right-boundary down set witness $(d, b)$ with $p \preceq
d$ then we need to do further work\footnote{The
case where $p \succ d$ will never occur in our algorithm, so we can ignore it.}. 
Note that the properties of a down set witness ensure that $d_2 \le f(d)_2$ and
$d_3 \le f(d)_3$. But there are two possibilities for dimension 1. If $d_1 \le
f(d)_1$ then $d \in \Up(f)$ and it can be returned by the inner algorithm. On
the other hand, if $d_1 > f(d)_1$, then we can show that $p$ and $d$ violate
order preservation.
We prove this formally in the following lemma.





\begin{lemma}
\label{lem:logn}
Let $L_{a, b}$ be a sub-instance that satisfies the invariant, and let $p$ be a
point satisfying $a \preceq p \preceq b$ that also satisfies one of the following conditions.
\begin{enumerate}
\item $p_1 = b_1$ and $p_1 < f(p)_1$. 
\item $p_2 = b_2$ and $p_2 < f(p)_2$. 
\item $p_1 = a_1$ and $p_1 > f(p)_1$. 
\item $p_2 = a_2$ and $p_2 > f(p)_2$. 
\end{enumerate}
Suppose further that, if there exists a down-set witness $(d, b)$ then $p
\preceq d$, and if there exists an up-set witness $(a, u)$ then $u \preceq p$. 
Then there exists a solution for the inner algorithm that can be found using
constantly many queries.
\end{lemma}

\myparagraph{\bf Initialization. }
The input to the algorithm is a sub-instance $L_{x, y}$, and recall that we have
fixed the principle slice $s = (\blank, \blank, s_3)$.  The initial values for
$a$ and $b$ are determined as follows. For each dimension $i$ we have
$a_i = s_3$ if $i = 3$, and $a_i = x_i$ otherwise, and we have $b_i = s_3$ if $i = 3$, and $b_i = y_i$ otherwise. 
That is, $a$ and $b$ are the projections of $x$ and $y$ onto $s$.

The following lemma, whose proof can be found in Appendix~\ref{app:init}, states
that either $a$ and~$b$ satisfy the invariant, or that we can
easily find a violation of order preservation. 

\begin{lemma}
\label{lem:init}
Either $L_{a, b}$ satisfies the invariant, or there is violation of order
preservation between $a$ and $x$, or between $b$ and $y$.
\end{lemma}

\myparagraph{\bf The algorithm.}
Now suppose that we have an instance $L_{a, b}$
that satisfies the invariant. We will describe how to execute one iteration of
the algorithm, which will either find a violation of order
preservation, or find a new instance whose size is at most half the
size of the $L_{a, b}$.

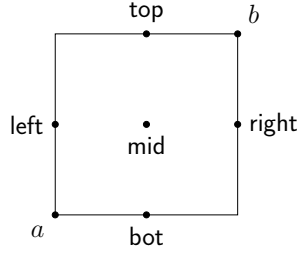
\begin{figure}
\begin{center}
\scalebox{0.4}{\begin{tikzpicture}[scale=0.6]

\tikzset{every node}=[font=\Huge]

\draw (0, 0) rectangle (10, 10);

\node[query,label=below left:{$a$}] at (0,0) {};
\node[query,label=above right:{$b$}] at (10,10) {};
	
\node[query,label=below:\midp] at (5,5) {};
\node[query,label=above:\topp] at (5,10) {};
\node[query,label=below:\botp] at (5,0) {};
\node[query,label=left:\leftp] at (0,5) {};
\node[query,label=right:\rightp] at (10,5) {};

\end{tikzpicture}
}
\end{center}
\caption{The five points used by the inner algorithm.
}
\label{fig:queries}
\end{figure}

We begin by defining some important points. 
We define $\midp = \lfloor (a + b)/2 \rfloor$ to be the \emph{midpoint}
of the instance, and we define 
the following points, 
which are shown in Figure~\ref{fig:queries}:
\begin{align*}
\botp   &= (\lfloor (a_1 + b_1)/2 \rfloor, a_2), &
\leftp  &= (a_1, \lfloor (a_2 + b_2)/2 \rfloor), \\
\topp   &= (\lfloor (a_1 + b_1)/2 \rfloor, b_2), &
\rightp &= (b_1, \lfloor (a_2 + b_2)/2 \rfloor). 
\end{align*}

\myparagraph{\bf Step 1: Fixing the up and down set witnesses.}
Suppose that $L_{a, b}$ satisfies the invariant with a top-boundary down set
witness $(d, b)$, We would like to ensure that $\topp \preceq d$, since
otherwise if we cut the instance in half in a later step,
the witness may no longer be within the sub-instance.
For the same reason, we would like to ensure 
that $(d, b)$ satisfies $\rightp \preceq d$ for a right-boundary down set witness, 
that $(a, u)$ satisfies $u \preceq \leftp$ for a left-boundary up set
witness, and
that $(a, u)$ satisfies $u \preceq \botp$ for a bottom-boundary up set witness.
By the end of Step 1 we will have either found a violation of order
preservation, moved into the next iteration with a sub-instance of half the
size, or all inequalities above will hold.

Step 1 consists of the following procedure. The procedure should be read
alongside Figure~\ref{fig:step1}, which gives a diagram for every case presented
below.

\begin{enumerate}

\item If $(d, b)$ is a top-boundary down set witness and $\topp \preceq d$ then there is no need to do anything. 
On the other hand, if 
$d \prec \topp$ we use the following procedure.

\begin{enumerate}
\item We first check if $\topp_3 > f(\topp)_3$. If this is the case, then since
the invariant ensures that $d_3 \le f(d)_3$ we have
$f(d)_3 \ge d_3 = \topp_3 > f(\topp)_3$
so $d \preceq \topp$ but $f(d) \not\preceq f(\topp)$, and an order preservation
violation has been found and the inner algorithm terminates.

\item We next check the whether $\topp_1 > f(\topp)_1$. In this case, we can
use $(d, \topp)$ as a down set witness for the sub-instance $L_{a,
\topp}$, where we observe that $\topp$ satisfies the requirements since 
$\topp_3 \le f(\topp)_3$ and $\topp_1 > f(\topp)_1$. Hence, $L_{a, \topp}$
satisfies the invariant (if $L_{a, b}$ also has an up set witness $(a, u)$ then
note that $u \preceq d$ continues to hold), and so the algorithm moves into the
next iteration with the sub-instance $L_{a, \topp}$.

\item In this final case we have $\topp_3 \le f(\topp)_3$ and $\topp_1 \le
f(\topp)_1$. Therefore $(\topp, b)$ is a valid down set witness for $L_{a, b}$
(if $L_{a, b}$ also has an up set witness $(a, u)$ then note that $u \preceq
d \prec \topp$). So we can replace $(d, b)$ with $(\topp, b)$ and continue,
noting that our down set witness now satisfies the required inequality.
\end{enumerate}

\item If $(d, b)$ is a right-boundary down set witness and $\rightp \preceq d$
then there is no need to do anything. On the other hand, if $d \prec \rightp$
then we use the same procedure as case 1, where dimensions 1 and 2 are exchanged
and the point \topp is replaced by the point \rightp.

\item 
If $(a, u)$ is a bottom-boundary up set witness and $u \preceq \botp$ then there
is no need to do anything. On the other hand, if $\botp \prec u$ then we use the
following procedure, which is the same as the procedure from case 1, where all
inequalities have been flipped.

\begin{enumerate}
\item We first check if $\botp_3 < f(\botp)_3$. If this is the case, then since
the invariant ensures that $u_3 \ge f(u)_3$ we have
$f(u)_3 \le u_3 = \botp_3 < f(\botp)_3$
so $u \succeq \botp$ but $f(u) \not\succeq f(\botp)$, and an order preservation
violation has been found and the inner algorithm terminates.

\item We next check the whether $\botp_1 < f(\botp)_1$. In this case, we can
use $(\botp, u)$ as an up set witness for the sub-instance $L_{\botp, b}$, where we observe that $\botp$ satisfies the requirements since 
$\botp_3 \ge f(\botp)_3$ and $\botp_1 < f(\botp)_1$. Hence, $L_{\botp, b}$
satisfies the invariant (if $L_{a, b}$ also has a down set witness $(d, b)$ then
note that $u \preceq d$ continues to hold), and so the algorithm moves into the
next iteration with the sub-instance $L_{\botp, b}$.

\item In this final case we have $\botp_3 \ge f(\botp)_3$ and $\botp_1 \ge
f(\botp)_1$. Therefore $(a, \botp)$ is a valid up set witness for $L_{a, b}$
(if $L_{a, b}$ also has a down set witness $(d, b)$ then we note that $\botp
\preceq u \preceq d$). So we can replace $(a, u)$ with $(a, \botp)$, noting that
our
up set witness now satisfies the required inequality.
\end{enumerate}

\item If $(a, u)$ is a left-boundary up set witness and $u \preceq \leftp$ 
then there is no need to do anything. On the other hand, if $u \succ \leftp$
then we use the same procedure as case 3, where dimensions 1 and 2 are exchanged
and the point \leftp is replaced by the point \botp.

\end{enumerate}

\begin{figure}
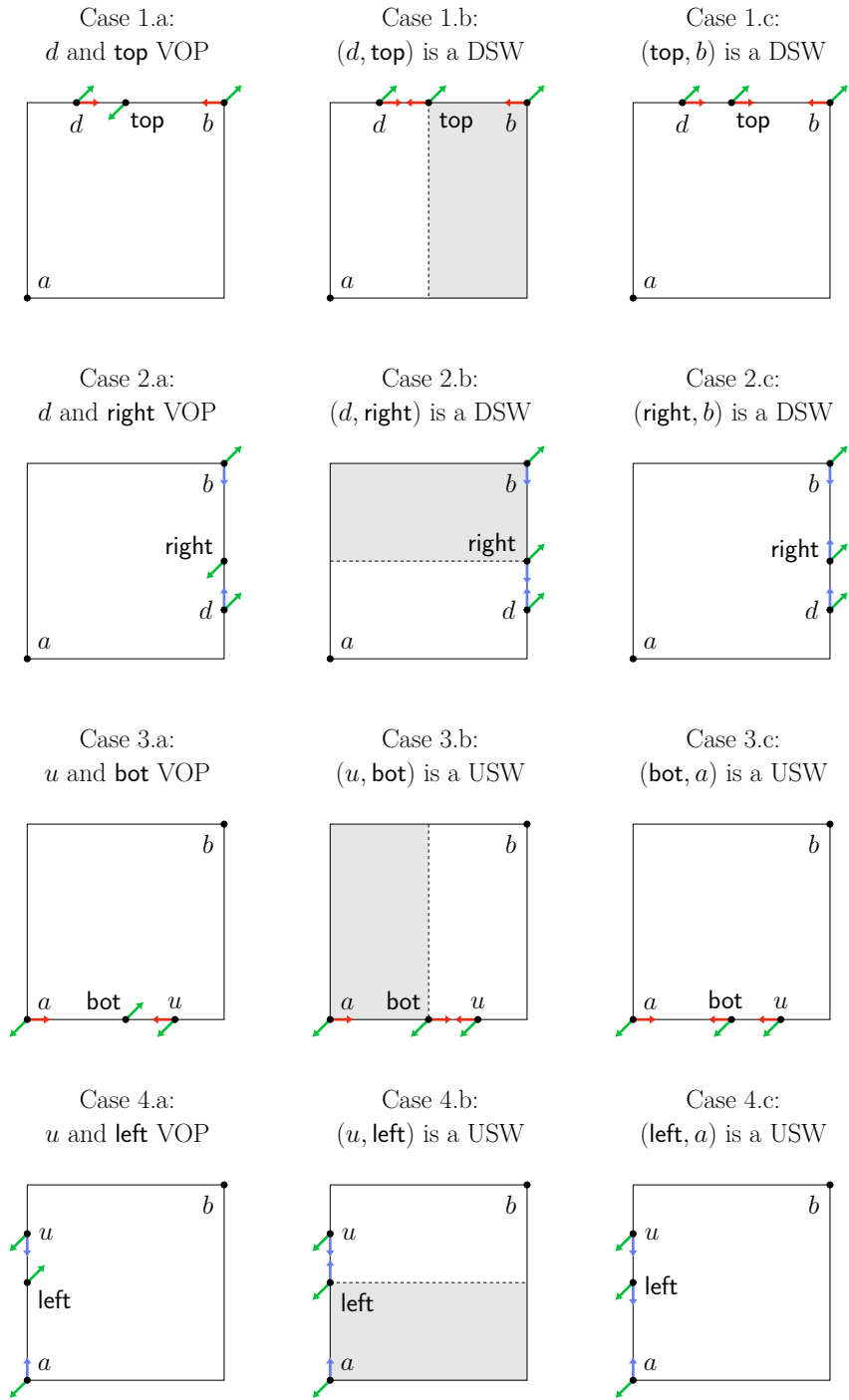

\centering
\scalebox{0.4}{ 
\input figures/step1.tex
}
\caption{All cases used in Step 1 of the algorithm. In the
labels, VOP is short for ``violate order preservation'', DSW is short for
``down set witness'', and USW is short for ``up set witness''.}
\label{fig:step1}
\end{figure}

\myparagraph{\bf Step 2: Find a smaller sub-instance.}
If Step 1 of the algorithm did not already move us into the next iteration of
the algorithm with a smaller instance, we apply Step 2. This step performs
a case analysis on the point \midp. 
The following procedure should be read in conjunction with 
Figure~\ref{fig:step2}, which provides a diagram for every case.

\begin{enumerate}
\item Check if $\midp_1 \le f(\midp)_1$ and $\midp_2 \le f(\midp)_2$. If this is
the case then $\midp \in \Up(f_s)$, and so we can move to the next iteration of
the algorithm with the sub-instance $L_{\midp, b}$. Note that if $L_{a,
b}$ has a down-set witness $(d, b)$, then Step 1 of the algorithm has ensured
that $\midp \preceq d$, and so $(d, b)$ is also a valid down-set witness for
$L_{\midp, b}$.

\item Check if $\midp_1 \ge f(\midp)_1$ and $\midp_2 \ge f(\midp)_2$. If this is
the case then $\midp \in \Down(f_s)$, and so we can move to the next iteration of
the algorithm with the sub-instance $L_{a, \midp}$. Note that if $L_{a,
b}$ has an up-set witness $(a, u)$, then Step 1 of the algorithm has ensured
that $u \preceq \midp$, and so $(a, u)$ is also a valid down-set witness for
$L_{a, \midp}$.

\item Check if $\midp_1 \le f(\midp)_1$ and $\midp_2 > f(\midp)_2$. If so, we use the following procedure.

\begin{enumerate}
\item Check if $\midp_3 \le f(\midp)_3$. If so, do the following.

\begin{enumerate}
\item Check if $\rightp_3 > f(\rightp)_3$. 
If this holds then we have 
$f(\midp)_3 \ge \midp_3 = \rightp_3 > f(\rightp)_3$,
meaning that $\midp \preceq \rightp$ but $f(\midp) \not\preceq f(\rightp)$. Thus
we have found a violation of order preservation and the algorithm terminates.

\item Check if $\rightp_1 < f(\rightp)_1$. 
If this holds then we use
Lemma~\ref{lem:logn} (with $p := \rightp$) to find a solution that can be
returned by the inner algorithm. 



\item If we reach this case then we have 
$\midp_3 \le f(\midp)_3$ and
$\rightp_3 \le f(\rightp)_3$, while we also have
$\midp_1 \le f(\midp)_1$ and $\rightp_1 \ge f(\rightp)_1$. Thus $(\midp, \rightp)$
is a valid down set witness for the instance $L_{a, \rightp}$. Note that if
$L_{a, b}$ has an up set witness $(a, u)$, then Step 1 ensures
that $u \preceq \midp$, and so $L_{a, \rightp}$ satisfies the invariant. 
So the algorithm moves to the next iteration with sub-instance 
$L_{a, \rightp}$.

\end{enumerate}

\item In this case we have $\midp_3 > f(\midp)_3$. The following three steps are
symmetric to those used in Case 3.a, but with all inequalities flipped,
dimension $1$ substituted for dimension $2$, the
point $\botp$ substituted for $\rightp$, and the point $a$ substituted for $b$.

\begin{enumerate}
\item Check if $\botp_3 > f(\botp)_3$. If this holds then we have 
$f(\midp)_3 \le \midp_3 = \botp_3 < f(\botp)_3$,
meaning that $\midp \succeq \botp$ but $f(\midp) \not\succeq f(\botp)$. Thus
we have found a violation of order preservation and the algorithm terminates.

\item Check if $\botp_2 > f(\botp)_2$. If this holds then we can use
Lemma~\ref{lem:logn} (with $p := \botp$) to find a solution that can be returned
by the inner algorithm.


\item If we reach this case then we have 
$\midp_3 \ge f(\midp)_3$ and
$\botp_3 \ge f(\botp)_3$, while we also have
$\midp_2 \ge f(\midp)_2$ and $\botp_2 \le f(\botp)_2$. Thus $(\botp, \midp)$
is a valid up set witness for the instance $L_{\botp, b}$. Note that if
$L_{a, b}$ has a down set witness $(d, b)$, then Step 1 of the algorithm ensures
that $d \succeq \midp$, and so $L_{\botp, b}$ satisfies the invariant. 
The algorithm will therefore move to the next iteration with the sub-instance 
$L_{\botp, b}$.

\end{enumerate}

\end{enumerate}

\item In this final case we have $\midp_1 > f(\midp)_1$ and $\midp_2 \le
f(\midp)_2$. Here we follow the same procedure as Case 3, but with dimensions 1
and 2 exchanged, every instance of the point $\rightp$ replaced with $\topp$,
and every instance of $\botp$ replaced with $\leftp$.

\end{enumerate}

\begin{figure}
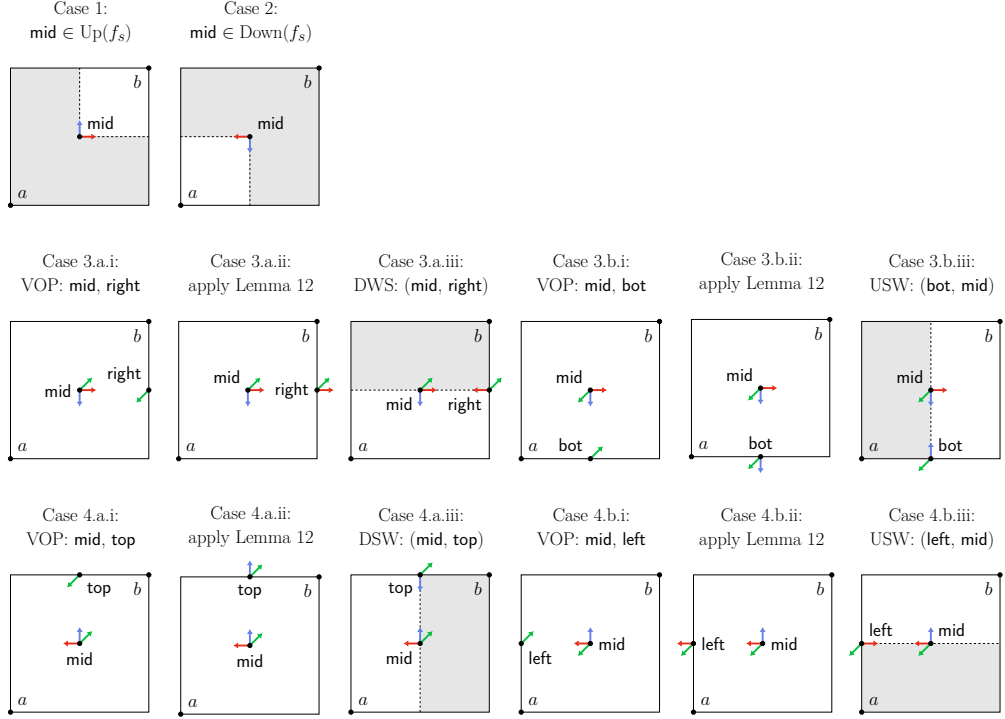

\centering
\scalebox{0.28}{ 
\input figures/step2.tex
}
\caption{All cases used in Step 2 of the algorithm. In the
labels, VOP is short for ``violate order preservation'', DSW is short for
``down set witness'', and USW is short for ``up set witness''.}
\label{fig:step2}
\end{figure}

\myparagraph{\bf The terminal phase of the algorithm.}
The algorithm can continue so long as $b_1 \ge a_1 + 2$ and $b_2 \ge a_2 + 2$,
since this ensures that all cases will cut the width of one of the
dimensions in half. 
The algorithm terminates once we have both $b_1 \le a_1 + 1$ \emph{and} $b_2 \le a_2 + 1$.
However, once there exists only one dimension $i$ for which $b_i \le a_i +
1$, we must be careful, since now the midpoint lies on the boundary of the
instance, and some of the cases of the algorithm may not rule out anything. We
deal with this scenario separately.

There are two distinct cases that we must deal with. The first case is a
\emph{width-one instance}, in which $b_i = a_i + 1$ for some index $i \in \{1,
2\}$, and $b_j > a_j + 1$ for the index $j \in \{1, 2\}$ with $j\ne i$, meaning
that the width of the shortest dimension is exactly one. These instances
are problematic because the midpoint \midp will now lie on the boundary of the
instance, and due to this, it is possible that the algorithm may be unable to
proceed.

We must also deal with \emph{width-zero instances}, in which $b_i = a_i$ for
some index $i \in \{1, 2\}$, and $b_j > a_j + 1$ for the index $j \in \{1, 2\}$
with $j\ne i$. These are one-dimensional subinstances, and once again it is
possible for the algorithm to be unable to proceed. 

We will use special procedures for width-one and width-zero instances, which we
outline below. 

\myparagraph{\bf Width-one instances.}
In the presentation below we will assume that the index $i = 1$, meaning that 
$b_1 = a_1 + 1$ (and hence the left-right width of the instance is one). The
case for $i = 2$ is symmetric. 

When the algorithm is presented with a width-one instance, it first performs
some preprocessing to ensure that there is no bottom-boundary up set witnesses, or
top-boundary down set witness. The preprocessing considers the following two
cases, which are shown in Figure~\ref{fig:width1_preprocess}.
\begin{enumerate}
\item If the instance has a bottom-boundary up set witness $(a, u)$, then note
that $a$ and $u$ are directly adjacent in dimension 1, and so 
Lemma~\ref{lem:upsetwitness} implies that we can either return $a$ or $u$ as a
solution for the inner algorithm, or that $a \in \Up(f_s)$, or $u \in \Up(f_s)$. 
\begin{enumerate}
\item If $a \in \Up(f_s)$, then the instance $L_{a, b}$ continues to satisfy the
invariant if we delete the up set witness.
\item If $u \in \Up(f_s)$, then the width-zero instance $L_{u, b}$ satisfies the
invariant, where we note that if $L_{a, b}$ has a down set witness $(d, b)$,
then since $u \preceq d$, we have that $(d, b)$ is also a valid down set witness
for $L_{u, b}$. 
\end{enumerate}

\item If the instance has a top-boundary down set witness $(d, b)$, then note
that $d$ and $b$ are directly adjacent in dimension 1, and so 
Lemma~\ref{lem:downsetwitness} implies that we can either return $d$ or $b$ as a
solution for the inner algorithm, or that $d \in \Down(f_s)$, or $b \in \Down(f_s)$. 
\begin{enumerate}
\item If $b \in \Down(f_s)$, then the instance $L_{a, b}$ continues to satisfy the
invariant if we delete the down set witness.
\item If $d \in \Down(f_s)$, then the width-zero instance $L_{a, d}$ satisfies the
invariant, where we note that if $L_{a, b}$ has an up set witness $(a, u)$,
then since $u \preceq d$, we have that $(a, u)$ is also a valid up set witness
for $L_{a, d}$. 
\end{enumerate}

\end{enumerate}

\begin{figure}
\centering
\scalebox{0.4}{ 


\begin{lrbox}{\rowonecolone}
\begin{tikzpicture}[scale=\myscale]
\tikzset{every node}=[font=\Huge]
\backarrow{0}{8}{white}
\backarrow{3}{8}{white}
\draw (0, 0) rectangle (3, 8);
\forwardarrow{0}{0}{darkpastelgreen}
\forwardarrow{3}{0}{darkpastelgreen}
\myrightarrow{0}{0}
\myleftarrow{3}{0}
\node[query,label={[yshift=2mm] left:$a$}] at (0,0) {};
\node[query,label={[yshift=2mm] right:$u$}] at (3,0) {};
\end{tikzpicture}
\end{lrbox}

\begin{lrbox}{\rowonecoltwo}
\begin{tikzpicture}[scale=\myscale]
\tikzset{every node}=[font=\Huge]
\forwardarrow{0}{0}{white}
\forwardarrow{3}{0}{white}
\draw (0, 0) rectangle (3, 8);
\backarrow{0}{8}{darkpastelgreen}
\backarrow{3}{8}{darkpastelgreen}
\myrightarrow{0}{8}
\myleftarrow{3}{8}
\node[query,label={[yshift=-2mm] left:$d$}] at (0,8) {};
\node[query,label={[yshift=-2mm] right:$b$}] at (3,8) {};
\end{tikzpicture}
\end{lrbox}

\begin{tikzpicture}

\subfigprime{0}{0}{Case 1:}{bottom-boundary USW}{\rowonecolone}
\subfigprime{10}{0}{Case 2:}{top-boundary DSW}{\rowonecoltwo}

\end{tikzpicture}
}
\caption{The two cases that trigger the preprocessing step for width-one
instances.}
\label{fig:width1_preprocess}
\end{figure}
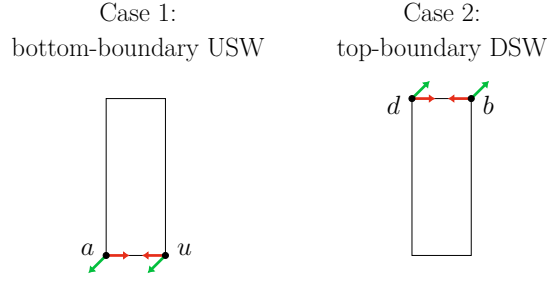

With the preprocessing completed, the algorithm uses two separate runs of
Steps~1 and~2, which each use a different midpoint.
In the first
run we use $\midpone = \lfloor (a + b)/2 \rfloor$ as normal, while in the second
run we use $\midptwo = \lceil (a + b)/2 \rceil$ as the midpoint, and we also
change the definitions of $\botp$, $\topp$, $\leftp$, and $\rightp$ to round up
instead of down. If either of the two runs decrease the size of the instance,
then we move to the next iteration on the smaller instance, where the reasoning
given in Steps 1 and 2 ensures that the instance continues to satisfy the
invariant. 

However, it could be the case that both runs do not decrease the size of the
instance. Due to the preprocessing, if Step 1 attempts to recurse on a smaller
sub-instance then it must succeed, since the only problematic cases are Case 1.b
and Case 3.b, which both depend on the existence of a top or bottom-boundary
witness, and the preprocessing ensures that these cannot exist.

On the other hand, Step 2 can fail to make progress in both runs. 
For the case where $i = 1$,
this can only occur if Case 3.b.iii of Step 2 triggered for the run with
$\midpone$ and Case 4.a.iii triggered for the run with $\midptwo$. But we can
argue that in this case a solution to the inner algorithm is easy to find.

Note that Case 3.b.iii can only trigger for $\midpone$ when $\midpone_1 \le
f(\midpone)_1$, while Case 4.a.iii can only trigger for $\midptwo$ when
$f(\midptwo)_1 < \midptwo$. Both of the following cases are shown in
Figure~\ref{fig:width1}.
\begin{enumerate}
\item If $\midpone_1 < f(\midpone)_1$ then we have
\begin{equation*}
f(\midptwo)_1 \le \midptwo_1 - 1 = \midpone_1 < f(\midpone)_1,
\end{equation*}
so we have $\midpone \preceq \midptwo$ but $f(\midpone) \not\preceq
f(\midptwo)$, meaning that $\midpone$ and $\midptwo$ witness a violation of
order preservation.

\item If $\midpone_1 = f(\midpone)_1$, then note that Case 3.b.iii ensures that
 $\midpone_2 \ge f(\midpone)_2$ and $\midpone_3 \ge f(\midpone)_3$. Therefore
$\midpone$ is in $\Down(f)$, so can be returned by the inner algorithm.

\end{enumerate}
So in both cases a solution to the inner algorithm has been found.


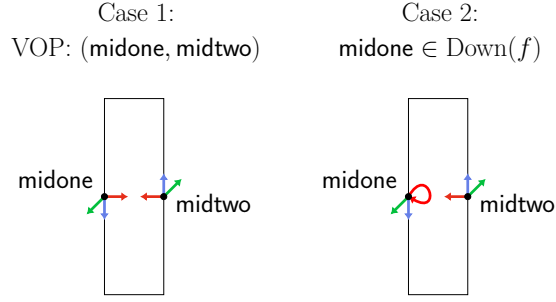
\begin{figure}
\centering
\scalebox{0.4}{ 

\tikzset{every loop/.style={min distance=20mm,in=-30,out=60}}

\begin{lrbox}{\rowonecolone}
\begin{tikzpicture}[scale=\myscale]
\tikzset{every node}=[font=\Huge]
\backarrow{0}{10}{white}
\backarrow{3}{10}{white}
\forwardarrow{0}{0}{white}
\forwardarrow{3}{0}{white}
\draw (0, 0) rectangle (3, 10);
\backarrow{3}{5}{darkpastelgreen}
\forwardarrow{0}{5}{darkpastelgreen}
\myrightarrow{0}{5}
\myleftarrow{3}{5}
\mydnarrow{0}{5}
\myuparrow{3}{5}
\node[query,label={[yshift=6mm] left:$\midpone$}] at (0,5) {};
\node[query,label={[yshift=-2mm] right:$\midptwo$}] at (3,5) {};
\end{tikzpicture}
\end{lrbox}

\begin{lrbox}{\rowonecoltwo}
\begin{tikzpicture}[scale=\myscale]
\tikzset{every node}=[font=\Huge]
\backarrow{0}{10}{white}
\backarrow{3}{10}{white}
\forwardarrow{0}{0}{white}
\forwardarrow{3}{0}{white}
\draw (0, 0) rectangle (3, 10);
\backarrow{3}{5}{darkpastelgreen}
\forwardarrow{0}{5}{darkpastelgreen}
\path[->, red, line width=\lw] (0,5) edge [loop above] (0,5);
\myleftarrow{3}{5}
\mydnarrow{0}{5}
\myuparrow{3}{5}
\node[query,label={[yshift=6mm] left:$\midpone$}] at (0,5) {};
\node[query,label={[yshift=-2mm] right:$\midptwo$}] at (3,5) {};
\end{tikzpicture}\end{lrbox}


\begin{tikzpicture}

\subfigprime{0}{0}{Case 1:}{VOP: $(\midpone, \midptwo)$}{\rowonecolone}
\subfigprime{10}{0}{Case 2:}{$\midpone \in \Down(f)$}{\rowonecoltwo}

\end{tikzpicture}
}
\caption{The two cases that are considered for width-one instances, when both
runs of the algorithm fail to make progress. In the left instance, $f(\midpone)$
strictly increases in dimension 1, while in the right instance $f(\midpone)$
does not move in dimension one, which we indicate with the self loop.}
\label{fig:width1}
\end{figure}

\myparagraph{\bf Width-zero instances.}
We again describe the procedure for the case where $i = 1$, meaning that the
instance has width zero in the left-right dimension. The case where $i = 2$ is
symmetric. 

The algorithm begins by performing a preprocessing step that removes any 
top-boundary down set witnesses or bottom-boundary up set witnesses. 
If there is a bottom-boundary up set witness $(a, u)$ then note that $a = u$,
and therefore Lemma~\ref{lem:upsetwitness} implies that either $a$ is a solution
that can be returned by the inner algorithm, or that $a \in \Up(f_s)$. Likewise,
if there is a top-boundary down set witness $(d, b)$, then $d = b$, and
Lemma~\ref{lem:downsetwitness} implies that either $d$ can be returned by the
inner algorithm, or $d \in \Down(f_s)$. 
Thus, the preprocessing step can either find a solution for the inner algorithm,
or produce an instance that satisfies the invariant that has no top-boundary
down set witness and no bottom-boundary up set witness.

Once the preprocessing has taken place, the algorithm proceeds through Step 1
and Step 2 as normal. If those steps make progress, then we continue on the
smaller width-zero instance. If they do not make progress, then we will show
that the inner algorithm can terminate after making at most $O(\log n)$ further
queries. 

We first observe that the only cases of Step 1 that would fail to make progress
are Case 1.b and Case 3.b, but neither of those cases can trigger because the
preprocessing step ensures that there is no bottom-boundary up set witness or
top-boundary down set witness.

On the other hand, Case 3.b.iii and Case 4.a.iii of Step 2 can fail to make
progress. We show how, in each of these cases, a solution for the inner
algorithm can be found by making at most $O(\log n)$ extra queries.
All of the following cases are depicted in Figure~\ref{fig:width0}.

\begin{enumerate}
\item If Case 3.b.iii is triggered, then note that $\midp_1 \le f(\midp)_1$. There are two cases to consider.
\begin{enumerate}
\item If $\midp_1 = f(\midp)_1$, then since $\midp_2 \ge f(\midp)_2$ and
$\midp_3 \ge f(\midp_3)$, we have that $\midp \in \Down(f)$, meaning that
$\midp$ can be returned by the inner algorithm.

\item If $\midp_1 < f(\midp)_1$ then we argue that an order preservation
violation can be found using at most $O(\log n)$ queries.
\begin{enumerate}
\item If $b \in \Down(f_s)$, then we have
\begin{equation*}
f(b)_1 \le b_1 = \midp_1 < f(\midp)_1,
\end{equation*}
meaning that $\midp \preceq b$, but $f(\midp) \not\preceq f(b)$, and so $\midp$
and $b$ violate order preservation.

\item If instead there is a down set witness $(d, b)$, then note that due to the
preprocessing, it must be a right-boundary down set witness, and due to Step 1,
we must have $\midp \preceq d$.  By Lemma~\ref{lem:downsetwitness} there exists
a point $x$ satisfying $d \preceq x \preceq b$ that can either be returned by
the inner algorithm, or that satisfies $x \in \Down(f_s)$. Furthermore, using
we can find this point in $O(\log n)$ queries using binary search. Thus we
can spend $O(\log n)$ queries and either immediatly terminate, or in the case where we find a point $x \in \Down(f_s)$, we can
repeat the argument above to show that $\midp$ and $x$ violate order
preservation.
\end{enumerate}
\end{enumerate}

\item If Case 4.a.iii is triggered, then note that $\midp_1 > f(\midp)_1$, and
we
argue that an order preservation
violation can be found using at most $O(\log n)$ queries.
\begin{enumerate}
\item If $a \in \Up(f_s)$, then we have
\begin{equation*}
f(\midp)_1 < \midp_1 = a_1 \le f(\midp)_1,
\end{equation*}
meaning that $a \preceq \midp$, but $f(a) \not\preceq f(\midp)$, and so $a$
and $\midp$ violate order preservation.

\item If instead there is an up set witness $(a, u)$, then note that due to the
preprocessing, it must be a left-boundary up set witness, and due to Step 1,
we must have $u \preceq \midp$.  By Lemma~\ref{lem:upsetwitness} there exists
a point $x$ satisfying $a \preceq x \preceq u$ that can either be returned by
the inner algorithm, or that satisfies $x \in \Up(f_s)$. Furthermore, we can
find this point in $O(\log n)$ queries using binary search. Thus we
can spend $O(\log n)$ queries and either immediatly terminate, or in the case
where we find a point $x \in \Up(f_s)$, we can repeat the argument above to show
that $\midp$ and $x$ violate order preservation.
\end{enumerate}
\end{enumerate}

\begin{figure}
\centering
\scalebox{0.4}{ 

\pgfmathsetmacro{\d}{2.5}
\pgfmathsetmacro{\dd}{10-\d}

\tikzset{every loop/.style={min distance=20mm,in=-30,out=60}}

\begin{lrbox}{\rowonecolone}
\begin{tikzpicture}[scale=\myscale]
\tikzset{every node}=[font=\Huge]
\backarrow{0}{12}{white}
\forwardarrow{0}{0}{white}
\draw (0, 0) -- (0, 12);
\path[->, red, line width=\lw] (0,6) edge [loop above] (0,6);
\mydnarrow{0}{6}
\forwardarrow{0}{6}{darkpastelgreen}
\node[query,label={[yshift=2mm] left:$\midp$}] at (0,6) {};
\end{tikzpicture}
\end{lrbox}

\begin{lrbox}{\rowonecoltwo}
\begin{tikzpicture}[scale=\myscale]
\tikzset{every node}=[font=\Huge]
\backarrow{0}{12}{white}
\forwardarrow{0}{0}{white}
\draw (0, 0) -- (0, 12);
\myleftarrow{0}{12}
\mydnarrow{0}{12}
\mydnarrow{0}{6}
\myrightarrow{0}{6}
\backarrow{0}{6}{darkpastelgreen}
\node[query,label=right:$b$] at (0,12) {};
\node[query,label={left:$\midp$}] at (0,6) {};
\end{tikzpicture}
\end{lrbox}

\begin{lrbox}{\rowonecolthr}
\begin{tikzpicture}[scale=\myscale]
\tikzset{every node}=[font=\Huge]
\forwardarrow{0}{0}{white}
\draw (0, 0) -- (0, 12);
\backarrow{0}{12}{darkpastelgreen}
\mydnarrow{0}{12}
\myuparrow{0}{9}
\backarrow{0}{9}{darkpastelgreen}
\mydnarrow{0}{6}
\myrightarrow{0}{6}
\backarrow{0}{6}{darkpastelgreen}
\node[query,label=left:$\midp$] at (0,6) {};
\node[query,label=left:$d$] at (0,9) {};
\node[query,label=left:$x$] at (0,10.5) {};
\node[query,label=left:$b$] at (0,12) {};
\end{tikzpicture}
\end{lrbox}

\newsavebox\rowonecolfou
\begin{lrbox}{\rowonecolfou}
\begin{tikzpicture}[scale=\myscale]
\tikzset{every node}=[font=\Huge]
\backarrow{0}{12}{white}
\forwardarrow{0}{0}{white}
\draw (0, 0) -- (0, 12);
\myrightarrow{0}{0}
\myuparrow{0}{0}
\myuparrow{0}{6}
\myleftarrow{0}{6}
\forwardarrow{0}{6}{darkpastelgreen}
\node[query,label=right:$\midp$] at (0,6) {};
\node[query,label=left:$a$] at (0,0) {};
\end{tikzpicture}
\end{lrbox}

\newsavebox\rowonecolfiv
\begin{lrbox}{\rowonecolfiv}
\begin{tikzpicture}[scale=\myscale]
\tikzset{every node}=[font=\Huge]
\backarrow{0}{12}{white}
\draw (0, 0) -- (0, 12);
\myuparrow{0}{6}
\myleftarrow{0}{6}
\forwardarrow{0}{6}{darkpastelgreen}
\forwardarrow{0}{0}{darkpastelgreen}
\myuparrow{0}{0}
\mydnarrow{0}{3}
\forwardarrow{0}{3}{darkpastelgreen}
\node[query,label=right:$a$] at (0,0) {};
\node[query,label=right:$u$] at (0,3) {};
\node[query,label=right:$\midp$] at (0,6) {};
\node[query,label=right:$x$] at (0,1.5) {};
\end{tikzpicture}
\end{lrbox}


\begin{tikzpicture}

\subfigprime{0}{0}{Case 1.a:}{$\midp \in \Down(f)$}{\rowonecolone}
\subfigprime{6}{0}{Case 1.b.i:}{VOP: $(\midp, b)$}{\rowonecoltwo}
\subfigprime{12}{0}{Case 1.b.ii:}{see caption}{\rowonecolthr}
\subfigprime{18}{0}{Case 2.a:}{VOP: $(a, \midp)$}{\rowonecolfou}
\subfigprime{24}{0}{Case 2.b:}{see caption}{\rowonecolfiv}

\end{tikzpicture}
}
\caption{The five cases that can be encountered if the algorithm fails to make
progress for a width-zero instance. In Cases 1.b.ii and 2.b we spend $O(\log n)$
queries to find the point $x$, which then allows us to terminate.}
\label{fig:width0}
\end{figure}
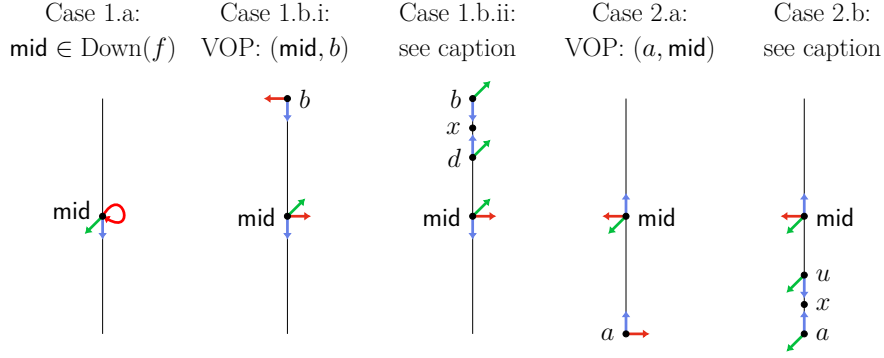

\myparagraph{\bf Termination.}
If the algorithm does not hit any of the cases that return a solution
immediately, then it will continue until it finds an instance $L_{a, b}$ with
$b_1 \le a_1 + 1$ and $b_2 \le a_2 + 1$ that satisfies the invariant.
Lemma~\ref{lem:invariant} implies that any sub-instance that satisfies the
invariant contains a solution that can be returned by the inner algorithm. Since
then $L_{a, b}$ contains at most four points, we can check all of them and then
return the solution that must exist.

\myparagraph{\bf Query complexity.}
Observe that each iteration of the algorithm either finds a violation of order
preservation, finds a solution after spending $O(\log n)$
further queries, or reduces the size of one of the dimensions by a factor of two.
Moreover, each non-terminating iteration of the algorithm queries at most five
points. 
Hence, if the algorithm is run on a sub-instance $L_{a, b}$ with $n_1 = b_1 -
a_1$ and $n_2 = b_2 - a_2$, then the algorithm will terminate after making at
most $O(\log n_1 +
\log n_2 + \log n)$ queries. So the overall query complexity of the algorithm is
$O(\log n)$, and we have shown the following theorem.

\begin{theorem}
\label{thm:inner}
There is an $O(\log n)$-query inner algorithm for 3-dimensional \tarski.
\end{theorem}

Theorems~\ref{thm:outer} and~\ref{thm:inner} imply that 3-dimensional \tarski
can be solved using $O(\log^2 n)$ queries, and this can be combined with the $\Omega(\log^2
n)$ lower bound for two-dimensional \tarski~\cite{EPRY20}, to give the following
theorem, where the straightforward lower bound is proved in Appendix~\ref{app:tight3d}.

\begin{theorem}
\label{thm:tight3d}
The deterministic query complexity of three-dimensional \tarski is $\Theta(\log^2 n)$.
\end{theorem}



\section{Higher Dimensions}

In this section we will show that 
$k$-dimensional \tarski can be solved using $O(\log^{2 \lceil k/3 \rceil} n)$
queries. We will prove this by showing that 
if $a$-dimensional \tarski can be solved in $q_a$ queries, and $b$-dimensional
\tarski can be solved in $q_b$ queries, then $(a\cdot b)$-dimensional \tarski
can be solved in $q_a \cdot (q_b + 2)$ queries. Using this fact along with our
new algorithm for three dimensional \tarski yields the result.

So, let $L = \lat(n_1, n_2, \dots, n_k)$ be a $k$-dimensional lattice. Observe that
for any pair of natural numbers $a, b \in \nats$ such that $a + b = k$ we have
that $L = A \times B$ where $A = \lat(n_1, n_2, \dots, n_a)$ and $B =
\lat(n_{a+1}, n_{a+2}, \dots, a_b)$. Given a point $x \in A$ and a point $y \in
B$, we will use $(x, y)$ to denote the point in $L$ that is obtained by
concatenating $x$ and $y$.

Suppose that we have an algorithm $\mathcal{A}$ that solves \tarski on $A$ using
$q_a$ queries, and an algorithm $\mathcal{B}$ that solves \tarski on $B$ using
$q_b$ queries. Our goal is to construct an algorithm $\mathcal{L}$ that solves
\tarski on $L$ using $q_a \cdot (q_b + 2)$ queries. 

\myparagraph{\bf The algorithm.}
At a high level, algorithm $\mathcal{L}$ will execute algorithm $\mathcal{A}$ on
the lattice $A$. We will use $f : L \rightarrow L$ to denote the \tarski instance on
$\mathcal{L}$, and $f_A : A \rightarrow A$ to denote the \tarski instance on
$\mathcal{A}$.

Every time $\mathcal{A}$ queries a point $x \in A$, algorithm $\mathcal{L}$ will fix
the following slice $s_x$ of $L$
\begin{equation*}
(s_x)_i = \begin{cases}
x_i & \text{if $i \le a$,} \\
\blank & \text{otherwise,}
\end{cases}
\end{equation*}
which is the slice defined by fixing the dimensions specified by $x$ and letting all other
dimensions be free. The slice $s_x$ defines a $b$-dimensional lattice $B$, and
so we can apply $\mathcal{B}$ to solve the \tarski problem on the slice $s_x$.
If $\mathcal{B}$ returns a violation of order preservation, then this is also a
violation of order preservation in $L$ and so $\mathcal{L}$ can terminate.
Otherwise, it will return a fixed point $y \in B$ of the slice $s_x$. Given $y$,
algorithm $\mathcal{L}$ can then respond to the query of $\mathcal{A}$ by
setting $f_A(x)_i = f((x, y))_i$ for each dimension $i \le a$. 

With this approach it is clear that if $\mathcal{A}$ returns a fixed point $x
\in A$, and if $y$ is the corresponding fixed point of $s_x$, then $(x, y)$
is a fixed point of $L$. 
However, if $\mathcal{A}$ returns a violation of order preservation, then the
approach so far is not sufficient to produce an order-preservation violation in
$L$. This is because, if $x_1 \preceq x_2$ violate order preservation
in $A$, then the pair of points $(x_1, y_1), (x_2, y_2) \in L$, where each $y_i$ is
the fixed point of $s_{x_i}$, do not necessarily violate order preservation,
because we have no guarantee that $(x_1, y_1) \preceq (x_2, y_2)$.

To address this, whenever we ask algorithm $\mathcal{B}$ to find a fixed point of the
slice $s_x$, we restrict it to a sub-instance of the slice to ensure that the
fixed point $y$ can later be used to witness order-preservation violations if
required. More concretely, the algorithm $\mathcal{L}$ does the following.
\begin{itemize}
\item Algorithm $\mathcal{L}$ maintains a set $P \subseteq A \times B$ of past
points, which is initially set so that $P := \emptyset$. Whenever algorithm $\mathcal{A}$
queries $x \in A$ and algorithm $\mathcal{B}$ responds with $y \in B$, the pair $(x, y)$
is added to $P$. 

\item Whenever $\mathcal{A}$ queries $x \in A$, algorithm $\mathcal{L}$ does the
following.
\begin{itemize}
\item It finds the set $D = \{ y' \; | \; \text{There exists an $x'$ such that }
(x', y') \in P \text{ and } x' \preceq x\}$, which contains the $y$ values
associated with all previously queried points that are below $x$ in the
$\preceq$ ordering. It then constructs the point $l \in B$ that is the least
upper bound of the set $D$.

\item Symmetrically, 
it finds the set $U = \{ y' \; | \; \text{There exists an $x'$ such that }
(x', y') \in P \text{ and } x \preceq x'\}$, which contains the $y$ values
associated with all previously queried points that are above $x$ in the
$\preceq$ ordering. It then constructs the point $u \in B$ that is the greatest
lower bound of the set $U$.
\end{itemize}

\item Then, if $B$ is the $b$-dimensional lattice defined by the slice $s_x$, 
algorithm $\mathcal{L}$ calls algorithm $\mathcal{B}$ on $B_{l, u}$, which is
the sub-instance of $B$ that contains all points between $l$ and $u$ in the
$\preceq$ ordering.
\end{itemize}
Other than this change, the algorithm proceeds as described earlier. 

\myparagraph{\bf Correctness.}
In order for the algorithm to work, 
we must show that the
sub-instance $B_{l, u}$ is non-empty, and contains a solution for the \tarski
problem. We verify both of these properties in the following pair of lemmas.
The first lemma shows that $B_{l, u}$ is not empty.

\begin{lemma}
In every iteration of the algorithm we have $l \preceq u$, which implies that 
$B_{l, u}$ is non-empty.
\end{lemma}
\begin{proof}
We begin by observing that, by construction, the algorithm ensures that for
every pair of points $(x_1, y_1), (x_2, y_2) \in P$, we have that $x_1 \preceq
x_2$ implies $y_1 \preceq y_2$. Hence, for every $y \in D$ and $y'
\in U$ we know that $y \preceq y'$. 

We can use this to prove that $l \preceq u$. Suppose, for the sake of
contradiction, that this is not the case, and therefore that there is some index
$i$ for which $l_i > u_i$. Let $y \in D$ be a point such that $y_i = l_i$,
which must exist since otherwise $l$ would not be the least upper bound of $D$.
Likewise, let $y' \in U$ be a point such that $y'_i = u_i$, which must exist
since otherwise $u$ would not be the greatest lower bound of $U$. We have
\begin{equation*}
y_i = l_i > u_i = y'_i,
\end{equation*}
and therefore we have $y \in D$, $y' \in U$, but $y \not\preceq y'$, giving us
our contradiction.
\end{proof}

The second lemma shows that $B_{l, u}$ contains a solution to the \tarski
problem. 

\begin{lemma}
\label{lem:luop}
Either a violation of order preservation in $L$ can be found between $l$ or $u$
with some previously queried point or the sub-instance $B_{l, u}$ contains a
solution to the \tarski problem.
\end{lemma}
\begin{proof}
We begin by showing that either $l$ is in the up set of $B$, or that a
violation of order preservation can be found. Suppose that there is some index
$i$ satisfying $a < i \le n$ such that $f(l)_i < l_i$. Then let $y \in D$ be a
point such that $y_i = l_i$, which must exist since otherwise $l$ would not be the
least upper bound of $D$. Moreover, since $l$ is the least upper bound of $D$,
we know that $y \preceq l$. 
Since $y$ was found in a previous iteration by
$\mathcal{B}$, we know that it is a fixed point in dimensions $a+1$ through $n$
and therefore $f(y)_i = y_i$. So we have
\begin{equation*}
f(y)_i = y_i = l_i > f(l)_i,
\end{equation*}
meaning that we have $y \preceq l$ but $f(y) \not\preceq f(l)$, and therefore
$y$ and $l$ violate order preservation.
So we have proved that either $y_i \le f(y)_i$ for all $i$ in $a < i \le n$, or
there exists an order-preservation violation between $l$ and some point $y$ in
$D$.

Symmetrically we can show that either $u$ is in the down set of $B$, or that a
violation of order preservation can be found. Suppose that there is some index
$i$ satisfying $a < i \le n$ such that $f(u)_i > u_i$. Then let $y \in U$ be a
point such that $y_i = u_i$, which must exist since otherwise $u$ would not be the
greatest lower bound of $U$. Moreover, since $u$ is the greatest lower bound of $U$,
we know that $u \preceq y$. 
Since $y$ was found in a previous iteration by
$\mathcal{B}$, we know that it is a fixed point in dimensions $a+1$ through $n$
and therefore $f(y)_i = y_i$. So we have
\begin{equation*}
f(u)_i > u_i = y_i = f(y)_i,
\end{equation*}
meaning that we have $u \preceq y$ but $f(u) \not\preceq f(y)$, and therefore
$u$ and $y$ violate order preservation.
So we have proved that either $u_i \ge f(y)_i$ for all $i$ in $a < i \le n$, or
there exists an order-preservation violation between $u$ and some point $y$ in
$U$.

Assuming that we have not found an order-preservation violation so far, we have
shown that $l$ is in the up set of $B$, and $u$ is in the down set of $B$.
Therefore we can apply Lemma~\ref{lem:tarskinew} to show that there must exist a
solution to the \tarski problem in $B_{l,u}$. 
\end{proof}

To complete the correctness argument we just need to observe that if algorithm
$\mathcal{A}$ returns a solution to the \tarski problem defined over $A$, then we
can find a solution to the \tarski problem defined over $L$.
\begin{itemize}
\item If $\mathcal{A}$ returns a point $x \in $ that is a fixed point of $A$, then the
point $(x, y)$, where $y$ was the fixed point returned by algorithm
$\mathcal{B}$
as the fixed point of $s_y$, is a fixed point of $L$.
\item If $\mathcal{A}$ returns two points $x_1, x_2 \in A$ that violate order
preservation in $A$, meaning that $x_1 \preceq x_2$ but $f_A(x_1) \not\preceq
f_A(x_2)$, then let $y_1$ and $y_2$ be the points found by algorithm
$\mathcal{B}$ for $x_1$ and $x_2$ respectively. Due to the procedure above, we
know that $x_1 \preceq x_2$ implies that $(x_1, y_1) \preceq (x_2,
y_2)$. Moreover, since there is some index $i$ in the range $1 \le i \le a$ such
that $f(x_1)_i > f(x_2)_i$, we have that $f(x_1, y_1)_i > f(x_2, y_2)_i$. Therefore,
the points $(x_1, y_1)$ and $(x_2, y_2)$ violate order preservation in $L$. 
\end{itemize}

\myparagraph{\bf Query Complexity.}
Recall that algorithm $\mathcal{A}$ solves \tarski using $q_a$ queries, and that
algorithm $\mathcal{B}$ solves \tarski using $q_b$ queries. In algorithm
$\mathcal{L}$, each time $\mathcal{A}$ makes a query we make queries to $l$ and
$u$, to verify that they do not violate order preservation in the sense of
Lemma~\ref{lem:luop}, and then make one call to $\mathcal{B}$. So in total
$\mathcal{L}$ makes $q_a \cdot (q_b + 2)$ queries, and we have shown the
following theorem.

\begin{theorem}
\label{thm:product}
If $a$-dimensional \tarski can be solved in $q_a$ queries, and $b$-dimensional
\tarski can be solved in $q_b$ queries, then $(a\cdot b)$-dimensional \tarski
can be solved in $q_a \cdot (q_b + 2)$ queries. 
\end{theorem}

We should remark that, although we have proved Theorem~\ref{thm:product} for
grid lattices, since those are the lattices that we study in this paper, a more
general version of the theorem for arbitrary lattices can also be shown to be true using
essentially the same techniques: if any lattice $L = A \times
B$, and there exist algorithms for \tarski for $A$ and $B$ that use $q_a$ and
$q_b$ queries, respectively, then \tarski can be solved on $L$ using $O(q_a
\cdot q_b)$ queries.

Theorem~\ref{thm:product} allows us to apply our $O(\log^2 n)$ query algorithm
for three-dimensional \tarski to obtain the following algorithm for $k$
dimensional \tarski.

\begin{theorem}
$k$-dimensional \tarski can be solved using $O(\log^{2 \lceil k/3 \rceil} n)$
queries. 
\end{theorem}
\begin{proof}
We can decompose the $k$-dimensional lattice $L$ into a product of 
lattices 
\begin{equation*}
L = L_1 \times L_2 \times \dots \times L_{\lceil k/3 \rceil},
\end{equation*}
where each
lattice $L_i$ has dimension at most three. Since we have 
an $O(\log^2 n)$ query algorithm for three-dimensional \tarski, we can
repeatedly apply Theorem~\ref{thm:product} to solve $k$-dimensional \tarski in 
$O(\log^{2 \lceil k/3 \rceil} n)$ queries.
\end{proof}

\myparagraph{\bf Time complexity.}
To obtain time complexity
results, note that writing down a point in the lattice $L$ already requires $k
\cdot \log n$ time. We assume that $f$ is implemented by a Boolean circuit of
size that is polynomial in $k$ and $\log n$. With this assumption, we observe
that all of our algorithms run in polynomial time with respect to $k$ and $\log
n$, and so our time
complexity result can be stated as follows.

\begin{theorem}
\label{thm:time}
If $f$ is presented as a Boolean circuit of size $\poly(\log n, k)$,
then 
there is an algorithm for \tarski that runs in time
$O(\poly(\log n, k) \cdot \log^{2 \lceil k/3 \rceil} n)$.
\end{theorem}







\section{Conclusion}

Our $O(\log^{2 \lceil k/3 \rceil} n)$ query algorithm for $k$-dimensional \tarski falsifies
prior conjectures that the problem required $\Omega(\log^k n)$
queries~\cite{DQY20,EPRY20}. This, of course, raises the question of what is
the query complexity of finding a Tarski fixed point? While our upper bound is tight
in three dimensions, it seems less likely to be the correct answer in higher
dimensions. Indeed, there seems to be a fairly wide range of possibilities. Is it
possible to show a $\log^{\Omega(k)} n$ query lower bound for the problem? Or perhaps there exists a
fixed parameter tractable algorithm that uses $O(f(k) \cdot \log^2 n)$ queries?
Both of those would be consistent with the known upper and lower bounds, and so
further research will be needed to close the gap.

Theorem~\ref{thm:product} provides a powerful new tool for reducing the query
complexity of the Tarski problem in high dimensions, since it allows a faster
algorithm for a constant-dimensional problem to reduce the query complexity of
the higher-dimensional problem. This means that determining the query complexity
of four- or five-dimensional Tarski problems could lead to even faster
algorithms for higher dimensions, and so further study of these problems seems
warranted. 


\newpage

\bibliography{references}

\appendix

\section{Proof of Lemma~\ref{lem:tarskinew}}
\label{app:tarskinew}

Before we prove Lemma~\ref{lem:tarskinew}, we first prove the following
auxiliary lemma. It states that if $x$ is in the up set, and $i$ is a dimension
for which $f$ moves strictly upwards, then either the point $x'$ that is
directly above $x$ in dimension $i$ is also in the upset, or $x$ and $x'$
witness a violation of the order preservation of $f$.

\begin{lemma}
\label{lem:upsetpath}
Let $L$ be a lattice, $f : L \rightarrow L$ be a \tarski instance, $x \in L$ be
a point satisfying $x \in \Up(f)$, and $i$ be a dimension such that $x_i <
f(x)_i$. Define the point $x'$ so that
\begin{equation*}
x'_j = \begin{cases}
x_j + 1 & \text{if $j = i$,} \\
x_j & \text{if $j \ne i$.}
\end{cases}
\end{equation*}
Either $x' \in \Up(f)$, or $x$ and $x'$ witness a violation of the order
preservation of $f$.
\end{lemma}
\begin{proof}
For dimension $i$, if $f(x')_i < x'_i$ then we have
\begin{equation*}
f(x')_i < x'_i \le f(x)_i,
\end{equation*}
where the second inequality follows from the assumption that $f(x)_i > x_i$, and
the fact that $x'_i = x_i + 1$. Hence, if $f(x')_i < x'_i$, then $x$ and $x'$
witness a violation of the order preservation of $f$ since $x \preceq x'$ but $x
\not\preceq f(x')$. 

For a dimension $j \ne i$, if $f(x')_j < x'_j$ then we have
\begin{equation*}
f(x')_j < x'_j = x_j \le f(x)_j,
\end{equation*}
where the equality follows from the fact that $x_j = x'_j$ by definition, and
the final inequality follows from the assumption that $x \in \Up(f)$. Hence, if 
$f(x')_j < x'_j$, then we have that $x \preceq x'$ while $f(x) \not\preceq
f(x')$, and so $x$ and $x'$ witness a violation of the order preservation of $f$.

From what we have proved above, we know that if $x$ and $x'$ do not violate
order preservation, then we have $x'_j \le f(x')_j$ for all dimensions $j$,
which implies that $x' \in \Up(f)$.
\end{proof}

We now prove Lemma~\ref{lem:tarskinew}.
Here we repeatedly apply Lemma~\ref{lem:upsetpath} to generate
a sequence of points in the up set that start at $a$ and monotonically increase
according to $\preceq$. If the path ends inside $L_{a, b}$ then we show that it
must end at a solution, while if it leaves $L_{a, b}$, then we show that
there is a violation of order preservation between $b$ and the point at
which the path leaves the sub-instance.

\begin{proof}
Since $a \in \Up(f)$, we have that either $a$ is a fixed point of $f$, or there
exists a dimension~$i$ such that $a_i < f(a)_i$.  In the former case we are
done, so let us assume that the latter is true. This means that we can apply
Lemma~\ref{lem:upsetpath} to obtain a point $a' = a + e_j$, where $e_j$ is the
unit vector in dimension $j$, such that either $a$ and $a'$ witness an order
preservation violation of $f$, or $a' \in \Up(f)$. 

By repeatedly applying the argument above, we can construct a sequence of points 
$$
a = a_1, a_2, a_3, \dots, a_p
$$
such that for all $i < p$ we have that 
and $a_{i} = a_{i-1} + e_j$
for some dimension $j$, and 
$a_i \in \Up(f)$.
The sequence ends at the point $a_p$ where we can no longer apply the argument,
which means that either $a_p$ and $a_{p-1}$ witness the order preservation
violation of $f$, or 
there is no index $j$ such that $(a_p)_j < f(a_p)_j$,
meaning that $a_p$ is a fixed point since $a_p \in \Up(f)$. Note that the sequence cannot be
infinite since $a_i \prec a_{i+1}$ for all $i$, and the lattice is finite.

If $a_p \preceq b$ then we are immediately done, since this implies that we have
one of the required solutions in the sub-instance $L_{a, b}$. 
On the other hand,
if $a_p \not\preceq b$, then let $i$ be the largest index such that $a_i \preceq
b$, and let $j$ be the dimension such that $a_{i+1} = a_i + e_j$. Note that we
have $b_j = (a_i)_j$, since otherwise we would have $a_{i+1} \preceq b$.
We have
\begin{equation*}
f(b)_j \le b_j = (a_i)_j < f(a_i)_j,
\end{equation*}
where the first inequality holds because $b \in \Down(f)$, and the final
inequality holds because $a_{i+1} = a_i + e_j$, which can only occur when
$(a_i)_j < f(a_i)_j$. Hence we have $a_i \preceq b$, but
$f(a_i) \not\preceq f(b)$, and so $a_i$ and $b$ witness a violation of the order
preservation of $f$. Furthermore, both $a_i$ and $b$ lie in the sub-instance
$L_{a, b}$, and so the proof is complete.

For the algorithm, note that the arguments above imply that it is sufficient to
find either $a_p$, or the index $i$ such that $a_i \preceq b$ and $a_{i+1} \not\preceq b$
($i$ must be unique if it exists, since $a_j \prec a_{j+1}$ for all $j$).
Each element of the sequence can be constructed by making a single query to $f$,
and since each step of the sequence strictly increases one coordinate of the
point, we have that there can be at most $\sum_{i = 1}^k (a_i - b_i)$ points
$a_i$ of the sequence satisfying $a_i \preceq b$. Thus the algorithm makes at
most

$O(\sum_{i = 1}^k (a_i - b_i))$ queries.
\end{proof}

\section{Proof of Lemma~\ref{lem:downsetwitness}}
\label{app:downsetwitness}

\begin{proof}
Observe that, since $d_i = b_i$, we have that $d$ and $b$ both lie in the same
one-dimensional slice. Moreover, the fact that $d_j \le f(d)_j$ implies that $d$
lies in the up set of this slice, while $f(b)_j \le b_j$ implies that $b$ lies
in the down set of this slice. Hence, we can apply Lemma~\ref{lem:tarskinew} to
either find a fixed point $p$ of this one-dimensional slice, or a violation of
order preservation. In the latter case, we are done since case two of the
lemma will have been fulfilled, so we will proceed assuming that $p$ exists.

There are now two cases to consider.
\begin{itemize}
\item \textbf{Case 1:} $\mathbf{p_i \le f(p)_i}$. 
There are two sub-cases.
\begin{itemize}
\item \textbf{Sub-case 1:} $\mathbf{p_3 \le f(p)_3}$. Note that $p_j = f(p)_j$, since $p$ is the
fixed point of the one-dimensional slice. Hence $p \in \Up(f)$, and so the first
case of the lemma has been fulfilled.

\item \textbf{Sub-case 2:} $\mathbf{p_3 > f(p)_3}$. Now we have
\begin{equation*}
f(p)_3 < p_3 = d_3 \le f(d)_3,
\end{equation*}
where the equality holds because $d$ and $p$ both lie in $s$, and the final
inequality holds by the requirements of a down set witness. Therefore we have $d
\preceq p$ but $f(d) \not\preceq f(p)$, and so $d$ and $p$ witness a violation
of the order preservation of $f$, and so the second case of the lemma is
satisfied.

\end{itemize}

\item \textbf{Case 2:} $p_i > f(p)_i$. Note that $p_j = f(p)_j$ since $p$ is a fixed
point of the one-dimensional slice. Hence $p$ is in the down set of the slice $s$,
and so the third condition of the lemma has been fulfilled.
\end{itemize}
\end{proof}

\section{Proof of Lemma~\ref{lem:upsetwitness}}
\label{app:upsetwitness}

\begin{proof}
This proof is exactly the same as the proof of Lemma~\ref{lem:downsetwitness},
but all inequalities have been flipped. We include it only for the sake of
completeness.

Observe that, since $a_i = u_i$, we have that $a$ and $u$ both lie in the same
one-dimensional slice. Moreover, the fact that $u_j \ge f(u)_j$ implies that $u$
lies in the down set of this slice, while $f(a)_j \ge a_j$ implies that $a$ lies
in the up set of this slice. Hence, we can apply Lemma~\ref{lem:tarskinew} to
either find a fixed point $p$ of this one-dimensional slice, or a violation of
order preservation. In the latter case, we are done since case two of the
lemma will have been fulfilled, so we will proceed assuming that $p$ exists.

There are now two cases to consider.
\begin{itemize}
\item \textbf{Case 1:} $\mathbf{p_i \ge f(p)_i}$. 
There are two sub-cases.
\begin{itemize}
\item \textbf{Sub-case 1:} $\mathbf{p_3 \ge f(p)_3}$. Note that $p_j = f(p)_j$,
	since $p$ is the fixed point of the one-dimensional slice. Hence $p \in
	\Down(f)$, and so the first case of the lemma has been fulfilled.

\item \textbf{Sub-case 2:} $\mathbf{p_3 < f(p)_3}$. Now we have
\begin{equation*}
f(p)_3 > p_3 = a_3 \ge f(a)_3,
\end{equation*}
where the equality holds because $a$ and $p$ both lie in $s$, and the final
inequality holds by the requirements of an up set witness. Therefore we have $a
\preceq p$ but $f(a) \not\preceq f(p)$, and so $a$ and $p$ witness a violation
of the order preservation of $f$, and so the second case of the lemma is
satisfied.

\end{itemize}

\item \textbf{Case 2:} $\mathbf{p_i < f(p)_i}$. Note that $p_j = f(p)_j$ since $p$ is a fixed
point of the one-dimensional slice. Hence $p$ is the up set of the slice $s$,
and so the third condition of the lemma has been fulfilled.
\end{itemize}
\end{proof}

\section{Proof of Lemma~\ref{lem:invariant}}
\label{app:invariant}

\begin{proof}
We can find a point $x \in L_{a, b}$ that satisfies $x \in \Up(f_s)$ in the
following way. By the invariant, either $a$ already satisfies this condition, or
we can apply Lemma~\ref{lem:upsetwitness} to obtain one of the three possible
cases from that lemma. Cases one and two immediately fulfill the requirements
of this lemma, and so we are done immediately in those cases, while the third
case gives us the point $x$.

Symmetrically, we can find a point $y \in L_{a, b}$ that satisfies $y \in
\Down(f_s)$, since either $b$ is such a point, or we can invoke
Lemma~\ref{lem:downsetwitness} to either immediately fulfil the requirements of
this lemma, or produce the point $y$.

Note further that we have $x \preceq y$. If either $a = x$ or $b = y$ then this
holds due to the promises given by the invariant. When we have both an up and
down set witness we have $x \preceq u \preceq d \preceq y$,
where the first and third inequalities are come from
Lemmas~\ref{lem:upsetwitness} and~\ref{lem:downsetwitness}, while $u \preceq d$
is promised by the invariant.

Hence we can apply Lemma~\ref{lem:tarskinew} to the sub-instance $L_{x, y}
\subseteq L_{a, b}$, which will either give us a violation of order
preservation, which will immediately satisfy the second condition of this lemma,
or a fixed point $p \in L_{x, y}$ of the slice $s$. We now do a case analysis on the third
dimension. 
\begin{itemize}
\item If $p_3 \le f(p)_3$, then $p \in \Up(f)$ since $p_1 = f(p_1)$ and $p_2 =
f(p_2)$. 
\item If $p_3 > f(p)_3$, then $p \in \Down(f)$ since $p_1 = f(p_1)$ and $p_2 =
f(p_2)$. 
\end{itemize}
Hence, in either case the first condition of this lemma is satisfied.
\end{proof}

\section{Proof of Lemma~\ref{lem:logn}}
\label{app:logn}

We shall begin by providing a proof for the case where $p_1 = b_1$ and $p_1 <
f(p)_1$. The other three cases will be proved symmetrically. There are two cases
to consider.

\begin{enumerate}
\item If $b \in \Down(f_s)$, or if there is a top-boundary down set witness $(d,
b)$, then we have that $f(b)_1 \le b_1$. Therefore we have
\begin{equation*}
f(b)_1 \le b_1 = p_1 < f(p)_1,
\end{equation*}
so $p \preceq b$ but $f(p) \not\preceq f(b)$, and therefore we have a violation
of order preservation.

\item If there is a right-boundary down set witness $(d, b)$, then note that by
the properties of a down set witness we have $d_2 \le f(d)_2$ and $d_3 \le
f(d)_3$. There are now two cases to consider.
\begin{enumerate}
\item If $d_1 \le f(d)_1$, then $d \in \Up(f)$, and so $d$ can be returned by
the inner algorithm.

\item If $d_1 > f(d)_1$ then we have 
\begin{equation*}
f(d)_1 < d_1 = p_1 < f(p)_1,
\end{equation*}
so we have $p \preceq d$ but $f(p) \not\preceq f(d)$, and therefore we have a
violation of order preservation.
\end{enumerate}
\end{enumerate}

The other three cases can be proved by the same reasoning.
\begin{itemize}

\item The case where $p_2 = b_2$ and $p_2 < f(p)_2$ can be proved by exchanging
dimensions 1 and 2.
\item The case where $p_1 = a_1$ and $p_1 > f(p)_1$ can be proved by flipping
all inequalities.
\item The case where $p_2 = a_2$ and $p_2 > f(p)_2$ can be proved by exchanging
dimensions 1 and 2, and also flipping all inequalities.

\end{itemize}

\section{Proof of Lemma~\ref{lem:init}}
\label{app:init}

\begin{proof}
We will show that either $a \in \Up(f_s)$ and $b \in \Down(f_s)$, or that a
violation of order preservation can be found.

We start by showing that either we have a violation of order preservation, or we
have that $a \in \Up(f_s)$. To check this, we only need
to inspect dimensions $i \in \{1, 2\}$. Note that if $f(a)_i < a_i$ for
some $i$, then 
\begin{equation*}
f(a)_i < a_i = x_i \le f(x)_i,
\end{equation*}
where the equality holds from the definition of $a$, and the final inequality
holds since $x \in \Up(f)$ is a requirement for calling the inner algorithm. Thus, if $f(a)_i < a_i$ for some dimension $i \in
\{1, 2\}$ then we have $x \preceq a$ and $f(x) \not\preceq f(a)$, and so we have that $x$ and $a$ witness a violation of the order preservation of $f$.
Otherwise, we have $f(a)_i \le a_i$ for all $i \in \{1, 2\}$, and therefore $a
\in \Up(f_s)$.

We can apply the same reasoning symmetrically to $b$ and $y$. If $b_i > f(b)_i$
for some $i \in \{1, 2\}$ then 
\begin{equation*}
f(b)_i > b_i = y_i \ge f(y)_i,
\end{equation*}
where the equality holds from the definition of $b$ and the final inequality
holds since $y \in \Down(f)$ is a requirement for calling the inner algorithm. Thus we would have $b \preceq y$ and $f(b)
\not\preceq f(y)$, and so either $b$ and $y$ witness a violation of the order
preservation of $f$, or $b \in \Down(f_s)$. 

At this stage we have either satisfied the second or third conditions of the
lemma, or we have that $a \in \Up(f_s)$ and $b \in \Down(f_s)$ and so the
invariant on $L_{a, b}$ is satisfied.
\end{proof}

\section{Proof of Theorem~\ref{thm:tight3d}}
\label{app:tight3d}

\begin{proof}
We will show that, if solving a $(k-1)$-dimensional \tarski instance requires
$q$ queries, then solving $k$-dimensional \tarski also requires $q$ queries.

Let $L^{k-1} = \lat(n_1, n_2, \dots, n_{k-1})$ be a $(k-1)$-dimensional lattice,
and let $f^{k-1} : L^{k-1} \rightarrow L^{k-1}$ be a \tarski instance over $L^{k-1}$
that requires $q$ queries to solve. Let $L^{k} = \lat(n_1, n_2, \dots, n_k)$ be
a $k$-dimensional lattice, where $n_k$ is any positive integer. We build the
function $f^{k} : L^{k} \rightarrow L^{k}$ in the following way. For each point $x \in L^{k}$ we define
\begin{equation*}
f^{k}(x)_i = \begin{cases}
f^{k-1}(x)_i & \text{if $i < k$,} \\
x_k & \text{if $i = k$.}
\end{cases}
\end{equation*}
Observe that $x = f^{k}(x)$ if and only if $x$ is also a fixed point of
$f^{k-1}$, and that $x$ and $y$ violate the order preservation of $f^{k}$ if and
only if $x$ and $y$ violate the order preservation of $f^{k-1}$. 
Moreover, every query to 
$f^{k}$ can be answered by making exactly one query to $f^{k-1}$. Hence, in
order solve the \tarski problem for $f^{k}$, we must make at least $q$ queries.

Thus, Theorem~\ref{thm:tight3d} follows from the $\Omega(\log^2 n)$ query lower 
bound of Etessami et al.~\cite{EPRY20} for \tarski in dimension 2. 
\end{proof}

\end{document}